\documentclass[twoside]{article}

\usepackage[accepted]{aistats2025}
\usepackage{multicol}
\usepackage{algorithm}
\usepackage[noend]{algpseudocode}
\usepackage{amsmath}
\usepackage{amsthm}
\usepackage{graphicx}
\usepackage{caption}
\usepackage{subcaption}
\usepackage{hyperref}
\usepackage[
    backend=biber,
    style=authoryear,    
]{biblatex}
\usepackage{amsfonts}
\usepackage{booktabs}
\newtheorem{theorem}{Theorem}
\newtheorem{lemma}{Lemma}
\newtheorem{definition}{Definition}
\newtheorem{proposition}{Proposition}

\newtheorem{corollary}{Corollary}

\DeclareCiteCommand{\cite}
  {\usebibmacro{prenote}}
  {\usebibmacro{citeindex}%
   \printtext[bibhyperref]{\usebibmacro{cite}}}
  {\multicitedelim}
  {\usebibmacro{postnote}}

\DeclareCiteCommand*{\cite}
  {\usebibmacro{prenote}}
  {\usebibmacro{citeindex}%
   \printtext[bibhyperref]{\usebibmacro{citeyear}}}
  {\multicitedelim}
  {\usebibmacro{postnote}}

\DeclareCiteCommand{\parencite}[\mkbibparens]
  {\usebibmacro{prenote}}
  {\usebibmacro{citeindex}%
    \printtext[bibhyperref]{\usebibmacro{cite}}}
  {\multicitedelim}
  {\usebibmacro{postnote}}

\DeclareCiteCommand*{\parencite}[\mkbibparens]
  {\usebibmacro{prenote}}
  {\usebibmacro{citeindex}%
    \printtext[bibhyperref]{\usebibmacro{citeyear}}}
  {\multicitedelim}
  {\usebibmacro{postnote}}

\DeclareCiteCommand{\footcite}[\mkbibfootnote]
  {\usebibmacro{prenote}}
  {\usebibmacro{citeindex}%
  \printtext[bibhyperref]{ \usebibmacro{cite}}}
  {\multicitedelim}
  {\usebibmacro{postnote}}

\DeclareCiteCommand{\footcitetext}[\mkbibfootnotetext]
  {\usebibmacro{prenote}}
  {\usebibmacro{citeindex}%
   \printtext[bibhyperref]{\usebibmacro{cite}}}
  {\multicitedelim}
  {\usebibmacro{postnote}}

\DeclareCiteCommand{\textcite}
  {\boolfalse{cbx:parens}}
  {\usebibmacro{citeindex}%
   \printtext[bibhyperref]{\usebibmacro{textcite}}}
  {\ifbool{cbx:parens}
     {\bibcloseparen\global\boolfalse{cbx:parens}}
     {}%
   \multicitedelim}
  {\usebibmacro{textcite:postnote}}

\begin{document}



\twocolumn[

\aistatstitle{Sampling from Multiscale Densities with Delayed Rejection Generalized Hamiltonian Monte Carlo}

\aistatsauthor{ Gilad Turok \And Chirag Modi \And  Bob Carpenter }

\aistatsaddress{Flatiron Institute \And  Flatiron Institute, New York University \And Flatiron Institute }
]

\begin{abstract}
  Hamiltonian Monte Carlo (HMC) is the mainstay of applied Bayesian inference for differentiable models. However, HMC still struggles to sample from hierarchical models that induce densities with multiscale geometry: a large step size is needed to efficiently explore low curvature regions while a small step size is needed to accurately explore high curvature regions. We introduce the delayed rejection generalized HMC (DR-G-HMC) sampler that overcomes this challenge by employing dynamic step size selection, inspired by differential equation solvers. In generalized HMC, each iteration does a single leapfrog step. DR-G-HMC sequentially makes proposals with geometrically decreasing step sizes upon rejection of earlier proposals. This simulates Hamiltonian dynamics that can adjust its step size along a (stochastic) Hamiltonian trajectory to deal with regions of high curvature. DR-G-HMC makes generalized HMC competitive by decreasing the number of rejections which otherwise cause inefficient backtracking and prevents directed movement. We present experiments to demonstrate that DR-G-HMC (1) correctly samples from multiscale densities, (2) makes generalized HMC methods competitive with the state of the art No-U-Turn sampler, and (3) is robust to tuning parameters.
\end{abstract}

\section{INTRODUCTION}

\begin{figure*}[!tbp]
  \begin{subfigure}[b]{0.63\textwidth}
    \includegraphics[width=\textwidth]{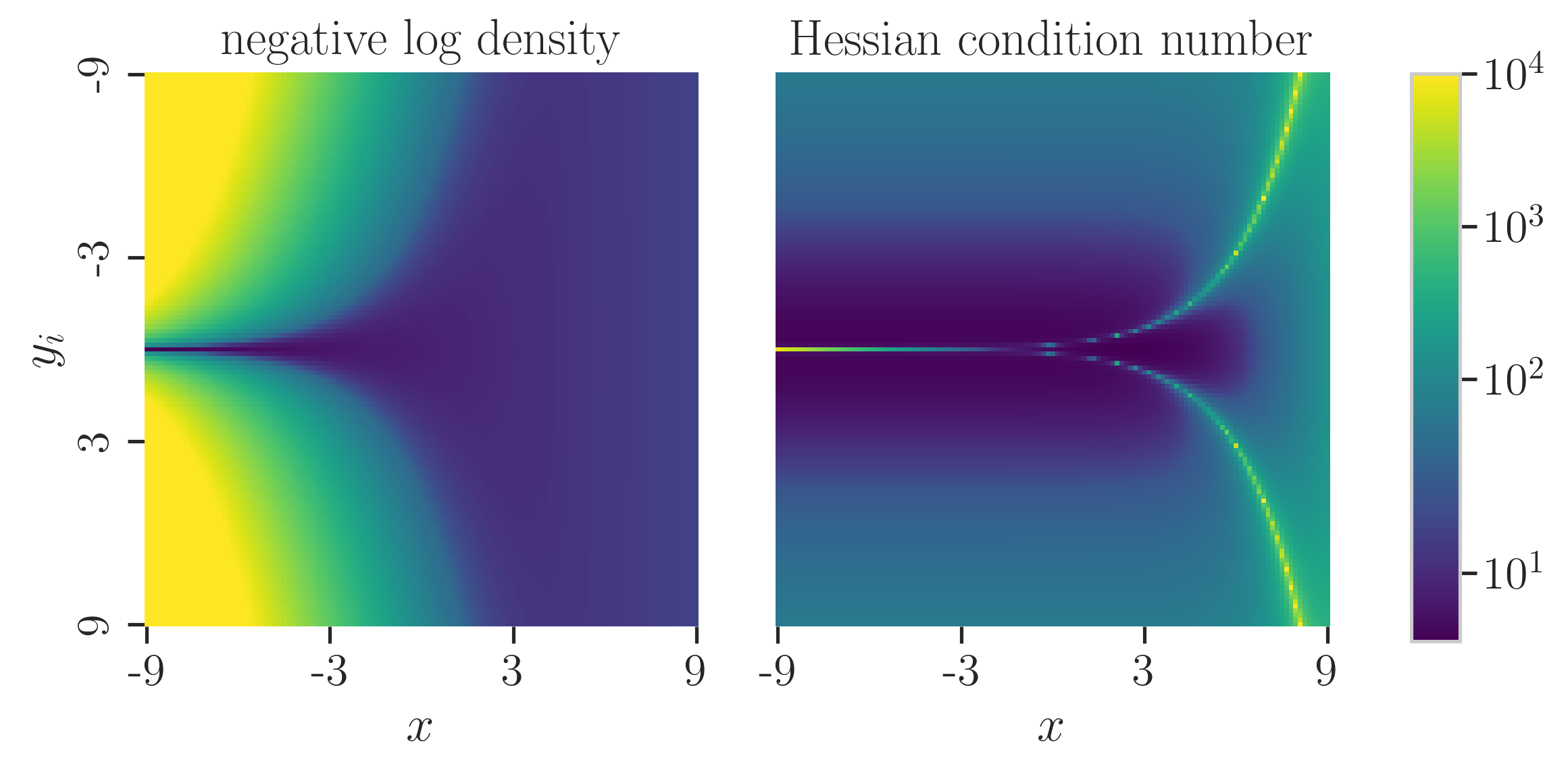}
    \caption{}
    \label{fig:funnel_multiscale}
  \end{subfigure}
  \hfill
  \begin{subfigure}[b]{0.35\textwidth}
    \includegraphics[width=\textwidth]{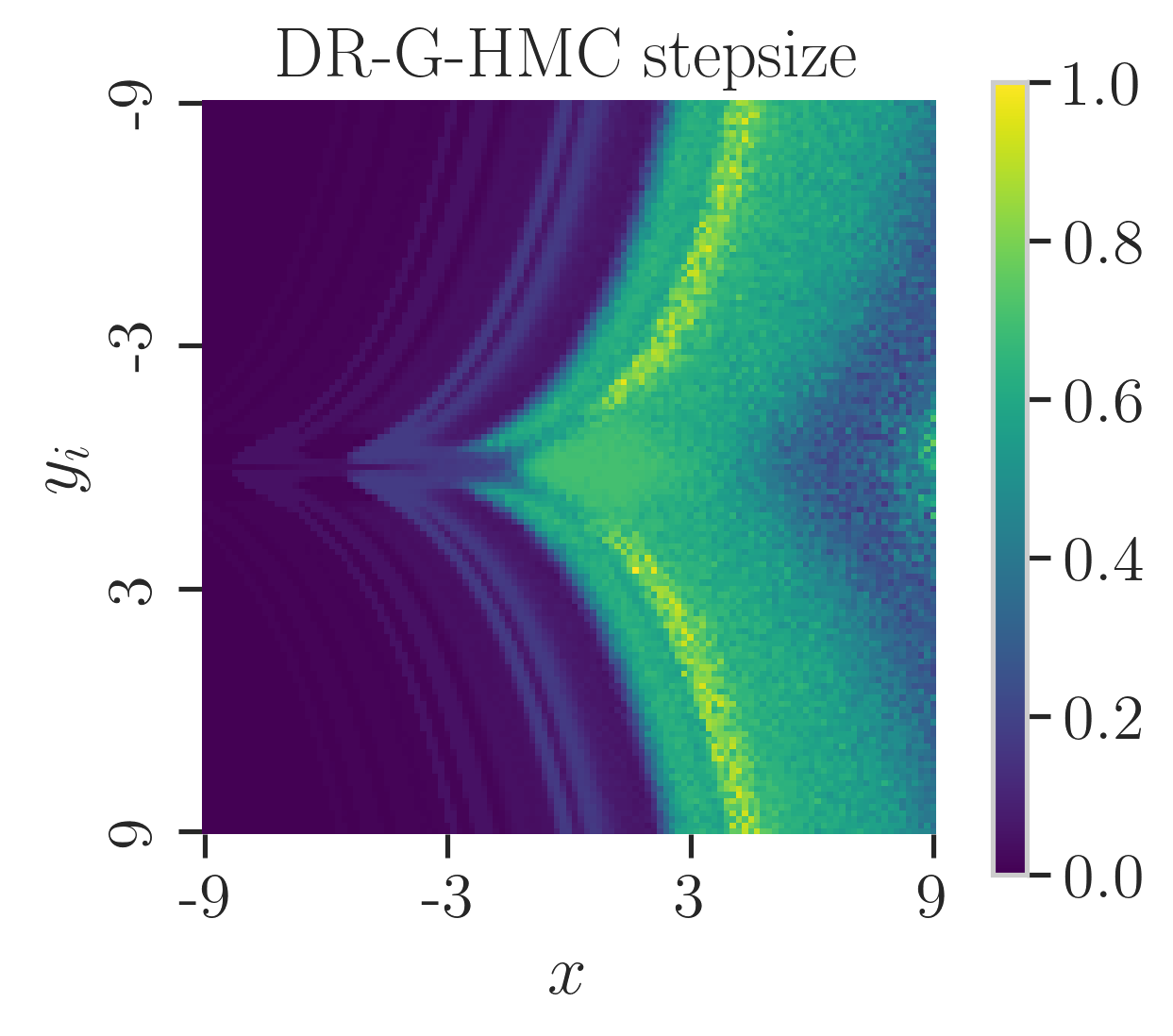}
    \caption{}
    \label{fig:funnel_drghmc_stepsize}
  \end{subfigure}
  \caption{\textbf{(a) Neal's funnel exhibits multiscale geometry.} Neal's funnel is challenging to sample from because its (negative) log density and Hessian condition number vary by orders of magnitude throughout the space. \textbf{(b) DR-G-HMC handles multiscale geometry with dynamic step sizes.} To sample from Neal's funnel, DR-G-HMC uses large step sizes in low-curvature regions ($x \gg 0$) and small step sizes in high curvature regions ($x \ll 0$). See \autoref{app:funnel_fig_details} for figure details.}
  \label{fig:funnel_intro}
\end{figure*}

Despite the success of Hamiltonian Monte Carlo (HMC) (\cite{duane1987hybrid}; \cite{neal2011mcmc}; \cite{betancourt2017conceptual}) and its auto-tuned variants like the No-U-Turn Sampler (NUTS) (\cite{hoffman2014no}), these algorithms struggle to sample from densities with multiscale geometry. Such densities are characterized by varying curvature in which scales can differ by orders of magnitude throughout the space. Multiscale geometry is challenging because we need small HMC step sizes in high-curvature regions for numerical stability, whereas we require large step sizes in flat regions for efficient sampling. This pathology commonly arises in hierarchical models (\cite{betancourt2015hamiltonian}; \cite{pourzanjani2019implicit}) in which small changes to a top-level parameter may induce large changes to parameters lower in the hierarchy. Neal's funnel is a density that exemplifies this multiscale behavior (see \autoref{fig:funnel_multiscale}). Neal's funnel contains a neck region, with high density and small volume, that opens to a mouth region, with low density and large volume. HMC methods with a fixed step size, like  NUTS, will struggle to sample from such densities.

\textcite{Modi_2023} proposed delayed rejection HMC (DR-HMC) to sample from multiscale densities. Instead of falling back to the original state upon rejection, delayed rejection methods make additional proposals in the \textit{same} iteration using a \textit{different} proposal kernel. If chosen appropriately, this new proposal kernel can lead to a better chance of acceptance. Motivated with time-step adaptive integrators for differential equations, DR-HMC makes proposals with increasingly smaller leapfrog step sizes, resulting in \textit{dynamic}, per-iteration step size selection. With these proposal kernels, DR-HMC can simulate more accurate Hamiltonian dynamics in high curvature regions, generating proposed states that are more likely to be accepted.

Although DR-HMC can successfully sample from multiscale densities, it introduces a few inefficiencies. Particularly, DR-HMC uses \textit{small} step sizes to move across the entirety of \textit{long} HMC trajectories, even if the small step size is necessary only for a short segment of the full trajectory. This makes DR-HMC iterations expensive, unable to compete with state of the art samplers like  NUTS. To remedy this issue, we introduce the delayed rejection generalized HMC (DR-G-HMC) sampler.

Generalized HMC (G-HMC) (\cite{horowitz1991generalized}) partially resamples the momentum variable and (typically) simulates Hamiltonian dynamics for a single leapfrog step. Unlike HMC, G-HMC poses a Metropolis accept/reject step after every leapfrog step which avoids long and expensive trajectories that may ultimately be rejected. Hence a delayed proposal in DR-G-HMC is made with a smaller step size only for the single leapfrog step \emph{where it is necessary}, unlike DR-HMC, which uses a smaller step size for the entire trajectory.

DR-G-HMC not only improves the computational efficiency of DR-HMC in robustly sampling multiscale distributions, it also overcomes longstanding inefficiencies in G-HMC. 
To maintain detailed balance in G-HMC, the momentum is flipped (negated) at the end of a sampling iteration. As a result, if a proposed state is rejected, a G-HMC chain reverses course. This results in less directed movement and increased autocorrelation for G-HMC samples, making G-HMC less efficient than HMC in practice. These rejections can be minimized by using an extremely small step size to more accurately simulate Hamiltonian dynamics, but this reduces the movement per-iteration and remains inefficient. 
DR-G-HMC provides an alternate way to minimize these rejections by making delayed proposals that increase the chances of acceptance. This minimizes the number of trajectory reversals while simultaneously being able to use large step sizes for the first proposals, making G-HMC competitive with state of the art samplers like  NUTS.

Our primary contributions are as follows:

\begin{itemize}
    \item We introduce the DR-G-HMC sampler that combines delayed rejection and generalized HMC while satisfying detailed balance.
    \item We demonstrate that DR-G-HMC robustly samples multiscale densities induced by hierarchical models, where it outperforms previous samplers like DR-HMC and NUTS.
    \item We show that DR-G-HMC solves the directed motion problem of generalized HMC, making it competitive with state of the art samplers like NUTS in sampling from densities without complex multiscale geometry.
    \item We present experiments to show that the performance of DR-G-HMC is insensitive to tuning different hyperparameters of the algorithm.
\end{itemize}

\section{RELATED WORK}

\subsection{Generalized HMC} G-HMC, introduced by \cite{horowitz1991generalized}, is less efficient than HMC due to frequent reversals upon rejections in the Markov chain Monte Carlo chain. \textcite{neal2020nonreversibly} proposed a scheme to cluster the acceptances together, allowing directed movement and hence higher efficiency.  This method was later auto-tuned by \textcite{MEADs} using many parallel chains. DR-G-HMC is an alternative approach to Neal's for fixing the directed motion problem of G-HMC, with the added benefits of multiscale sampling and insensitivity to tuning parameters.

\subsection{Delayed Rejection}  Delayed rejection methods have been studied for Metropolis-Hastings algorithms (\cite{tierney1999some}; \cite{green2001delayed}; \cite{haario2006dram}; \cite{bedard2014scaling}) but have only recently been adapted to HMC samplers. \textcite{Modi_2023} apply delayed rejection to HMC in the DR-HMC sampler to robustly sample from multiscale densities. Although DR-HMC and DR-G-HMC share similarities, DR-G-HMC makes notable efficiency gains by using a smaller step size only where necessary. \textcite{sohl2014hamiltonian} and \textcite{campos2015extra} present a different flavor of delayed rejection with the goal of increasing the autocorrelation length in G-HMC. These samplers \textit{continue} the trajectory forward upon rejection with the same proposal kernel, which will fail if rejection is due to change in curvature (i.e., a diverging Hamiltonian). In contrast, we apply delayed rejection to G-HMC by \textit{retrying} from the same starting point with a different proposal kernel. This key difference allows our method to sample from multiscale densities.

\subsection{Step-size Selection} Several approaches have been proposed for step-size selection in Markov chain Monte Carlo samplers (theoretical: \cite{atchade2006adaptive}; \cite{marshall2012adaptive}; \cite{gelman1997weak}, adaptive:  \cite{hoffman2014no}, empirical: \cite{coullon2023efficient}; \cite{campbell2021gradient}). These approaches all struggle with multiscale densities because they fix a \textit{single} step size. Other work selects \textit{dynamic} step sizes that adapt to local curvature at every sampling iteration. Previous approaches include generalizing from Euclidean space to Riemmanian manifolds to explicitly deal with curvature (\cite{girolami2011riemann}; \cite{betancourt2015hamiltonian}) and applying implicit symplectic integrators that can adapt to the stiffness induced by multiscale densities (\cite{pourzanjani2019implicit}; \cite{brofos2021evaluating}).\footnote{Diagonal metrics are equivalent to a per-dimension step size; dense metrics account for curvature.} However both methods are computationally expensive, require implicit integrators that are hard to tune, and the former requires a Riemannian (i.e., positive definite) metric such as Fisher information.  Cheaper dynamic step size selection is performed in \textcite{kleppe2016adaptive} by bounding the error in the Hamiltonian and in \textcite{autoMALA} by sampling from a step-size distribution that maintains reversibility.

\section{BACKGROUND}\label{sec:background}

We seek to generate samples from a differentiable and potentially unnormalized probability density function $\pi(\theta)$ with parameters $\theta \in \mathbb{R}^D$. The target density $\pi$ is often a Bayesian posterior.

\subsection{Hamiltonian Monte Carlo} HMC generates samples by treating $\theta$ as position vector and introducing an auxiliary momentum vector $\rho \in \mathbb{R}^D$ for a fictitious particle with mass matrix $M$. The position and momentum define a Hamiltonian $$H(\theta, \rho) =  - \log \pi (\theta) + \frac{1}{2} \rho^T M^{-1} \rho$$ with corresponding Gibbs density 
\begin{align}\label{eq:gibbs}
    \tilde{\pi}(\theta, \rho) &\propto \exp(-H(\theta, \rho)) 
    \\ & = \pi(\theta) \cdot \text{normal}(\rho \mid 0, M), \notag
\end{align}
for some symmetric, positive-definite mass matrix $M$. HMC generates samples $(\theta, \rho)$ in this augmented state space as a Metropolis-within-Gibbs sampler: a Gibbs step updates the momentum $\rho$ and a Metropolis-Hastings step updates the position $\theta$. From samples $(\theta, \rho)$ we can easily recover samples $\theta$ from the target density $\pi$ by extracting the first $D$ coordinates.

To generate the next state from the current state $(\theta, \rho)$, first Gibbs resample the momentum,  $\rho^\prime \sim \textrm{normal}(0, M)$. Then perform the Metropolis update as follows: generate a proposed state by simulating Hamiltonian dynamics, approximated with leapfrog integration, then flip (negate) the momentum. To do so, define the proposal map $F = P L_\epsilon^n$ as the composition of a momentum flip $P(\theta, \rho) = (\theta, -\rho)$ and $n$ leapfrog steps $L_\epsilon^n$ with step size $\epsilon$. The proposed state $(\theta^{\mathrm{pr}}, \rho^{\mathrm{pr}}) = F(\theta, \rho^\prime)$ is accepted with probability 
$$\alpha = \min\!\left(1, \frac{\pi(\theta^{\mathrm{pr}}) \cdot \textrm{normal}(\rho^{\mathrm{pr}} \mid 0, M)}{\pi(\theta) \cdot \textrm{normal}(\rho^\prime \mid 0, M)} \right),$$  otherwise we remain at the current state $(\theta, \rho^\prime)$.

\subsection{Generalized Hamiltonian Monte Carlo} Instead of completely resampling the momentum as $\rho^\prime \sim \textrm{normal}(0, M)$, G-HMC partially updates the momentum as $\rho^\prime \sim \textrm{normal}(\rho \sqrt{1 - \gamma }, \gamma M)$ for some damping parameter $\gamma \in (0, 1]$ (\cite{horowitz1991generalized}). The damping $\gamma$ controls how much the momentum changes between iterations, while leaving the distribution over $\textrm{normal}(0, M)$ invariant. 
% The damping $\gamma$ is analogous to the trajectory length in HMC.
Next, G-HMC generates a proposed state $(\theta^{\mathrm{pr}}, \rho^{\mathrm{pr}})$ with proposal map $F$ and accepts it with some probability $\alpha$.  Unlike HMC, most G-HMC algorithms use one leapfrog step per iteration ($n=1$), with partial momentum refreshment promoting directed movement. For the rest of this work, we will fix $n=1$ as it allows for immediate feedback when exploring a density. Lastly at the end of a G-HMC iteration the momentum is unconditionally negated. If the proposed state $(\theta^{\mathrm{pr}}, \rho^{\mathrm{pr}})$ is accepted, this undoes the negation from the proposal map $F$; if rejected, this negation makes the chain reverse direction. Thus G-HMC struggles to make directed motion because of inefficient reversals upon rejection. Standard HMC is unaffected by momentum negations because the momentum $\rho$ is completely resampled at every iteration.

\section{DELAYED REJECTION GENERALIZED HAMILTONIAN MONTE CARLO}\label{sec:drghmc}

\begin{algorithm*}[t]
	\caption{Delayed rejection generalized Hamiltonian Monte Carlo sampler} 
        \label{alg:DR-G-HMC}
	\begin{algorithmic}  
        \begin{multicols}{2}
        [\Require target density $\pi$, max proposals $K$, damping $\gamma$, initial step size $\epsilon$, reduction factor $r$, mass matrix $M$]

		\Function{sample}{$\theta, \rho$}
		\State $\rho^\prime \sim \textrm{normal}(\rho \sqrt{1 - \gamma}, \gamma M)$
		\For{$k$ in $1 \ldots K$}
		\State $\theta^{\mathrm{pr}}, \rho^{\mathrm{pr}} \leftarrow \textproc{propose}(\theta, \rho^\prime, k)$
		\State $ \alpha \leftarrow \textproc{accept}(\theta, \rho^\prime, \theta^{\mathrm{pr}}, \rho^{\mathrm{pr}}, k)$
		\State $u \sim \textrm{uniform}(0, 1)$
		\If{$u < \alpha$}
            \State \Return{$\textproc{flip}(\theta^{\mathrm{pr}}, \rho^{\mathrm{pr}})$}
		\EndIf
		\EndFor
            \State \Return{$\textproc{flip}(\theta, \rho^\prime)$}
		\EndFunction
            \\
             \Function{propose}{$\theta, \rho, k$}
            \State $\epsilon_k \leftarrow \epsilon / r ^ {k-1}$
		\State $\theta, \rho \leftarrow \textproc{leapfrog}(\theta, \rho, \epsilon_k)$
		\State \Return{$\textproc{flip}(\theta, \rho)$}
		\EndFunction
            \\
            \Function{flip}{$\theta, \rho$}
            \State \Return{$\theta, -\rho$}
            \EndFunction
            \columnbreak
            \Function{leapfrog}{$\theta, \rho, \epsilon$}
		\State $\rho^\prime \leftarrow \rho + \frac{\epsilon}{2} \nabla \log \pi(\theta)$
		\State $\theta^\prime \leftarrow \theta + \epsilon M^{-1} \rho^\prime$
		\State $\rho^{\prime \prime} \leftarrow \rho^\prime + \frac{\epsilon}{2} \nabla \log \pi (\theta^\prime)$
		\State \Return{$\theta^\prime, \rho^{\prime \prime}$}
		\EndFunction
            \\
            \Function{accept}{$\theta, \rho^\prime, \theta^{\mathrm{pr}}, \rho^{\mathrm{pr}}, k$}
		\State $\alpha \leftarrow \frac{\displaystyle \pi(\theta^{\mathrm{pr}}) \cdot \textrm{normal}(\rho^{\mathrm{pr}} \mid 0, M)}{\displaystyle \pi(\theta) \cdot \textrm{normal} (\rho^\prime \mid 0, M)}$
		\For {$i$ in $1 \ldots k - 1$}
		\State $\theta^{\mathrm{gh}}, \rho^{\mathrm{gh}} \leftarrow \textproc{propose}(\theta^{\mathrm{pr}}, \rho^{\mathrm{pr}}, i)$
		\State $\theta^{\mathrm{rej}}, \rho^{\mathrm{rej}} \leftarrow \textproc{propose}(\theta, \rho^\prime, i)$ \Comment{cached}
		\State $ \alpha \leftarrow \alpha \cdot \dfrac{ 1 - \textproc{accept}(\theta^{\mathrm{pr}}, \rho^{\mathrm{pr}}, \theta^{\mathrm{gh}}, \rho^{\mathrm{gh}}, i)}{ 1 - \textproc{accept}(\theta, \rho^\prime, \theta^{\mathrm{rej}}, \rho^{\mathrm{rej}}, i)} $
		\EndFor
		\State \Return{$\min(1, \alpha)$}
		\EndFunction
        \end{multicols}
	\end{algorithmic}
\end{algorithm*}

DR-G-HMC applies delayed rejection to G-HMC to allow smaller step sizes to be used if a larger step size proposal is rejected during a single sampling iteration. The pseudocode is provided in \autoref{alg:DR-G-HMC} and is briefly described here. More technical details, the proof of invariance over $\pi$, and the derivation of acceptance probability are presented in \autoref{app:drghmc}.

Every iteration of DR-G-HMC begins with a Gibbs update by partially updating the momentum as $\rho^\prime \sim \textrm{normal}(\rho \sqrt{1 - \gamma }, \gamma M)$, same as G-HMC. Then, the sampler can make up to a maximum of $K$ Metropolis-Hastings updates where each subsequent proposal is made only if the previous one is rejected.

For proposal $k$, we define a proposal map $F_k = P L_{\epsilon_k}^{n_k}$ with its own leapfrog step size $\epsilon_k$ and set the number of steps $n_k=1$, following G-HMC. The acceptance probability $\alpha_k$ of the $k$th proposed state $(\theta^{\mathrm{pr}}, \rho^{\mathrm{pr}}) = F_k(\theta, \rho^\prime)$ is evaluated to account for rejecting the previous $k-1$ proposals and is discussed below. If accepted, $(\theta^{\mathrm{pr}}, \rho^{\mathrm{pr}})$ becomes the next state of the Markov chain. Otherwise we delay rejection and make another proposal. If a maximum of $K$ proposals are made, we finally reject and set $(\theta, \rho^\prime)$ as the next state of the Markov chain. Lastly, the momentum is flipped regardless of whether any of the $K$ proposals are accepted or rejected in order to maintain reversibility.

To maintain detailed balance, the acceptance probability $\alpha_k$ for accepting or rejecting the $k$th proposed state 
$(\theta^{\mathrm{pr}}, \rho^{\mathrm{pr}}) = F_k(\theta, \rho^\prime)$ is given as
\begin{multline} \label{eq:accept_prob}
    \alpha_k(\theta, \rho^\prime, F_k(\theta, \rho^\prime)) \\
    {} = \min \bigg( 1, \frac{\pi( \theta^{\mathrm{pr}}) \cdot \textrm{normal}(\rho^{\mathrm{pr}} \mid 0, M)}{\pi(\theta) \cdot \textrm{normal}(\rho^\prime \mid 0, M)} \\
    \times \prod_{i=1}^{k-1}\frac{ 1 - \alpha_i \big( \theta^{\mathrm{pr}}, \rho^{\mathrm{pr}}, F_i(\theta^{\mathrm{pr}}, \rho^{\mathrm{pr}}) \big)}{ 1 - \alpha_i \big( \theta, \rho^\prime, F_i(\theta, \rho^\prime) \big)} \bigg).
\end{multline}
The full derivation of this acceptance probability is given in \autoref{app:accept_prob}, but we give a brief intuition here. The additional factor in $\alpha_k$ over the standard HMC acceptance probability is the product over $k$ terms. The denominator in this product is the probability of  rejecting the previous $k-1$ proposals originating from the current state $(\theta^{\mathrm{rej}}, \rho^{\mathrm{rej}}) = F_i(\theta, \rho^\prime)$. These terms have already been computed in previous proposals and are stored in a stack data structure, so they do not require additional computation. To maintain detailed balance, we need to balance it against the similar probability of rejecting the first $k-1$ proposals originating from the proposed state $(\theta^{\mathrm{gh}}, \rho^{\mathrm{gh}}) = F_i(\theta^{\mathrm{pr}}, \rho^{\mathrm{pr}})$. These are the terms in the numerator and require additional computation for every additional proposal. We call these ``ghost states'' and write $(\theta^{\mathrm{gh}}, \rho^{\mathrm{gh}})$, because they are the $i$-th proposal that would have been made and rejected in a hypothetical chain that had \textit{started} from proposed state $(\theta^{\mathrm{pr}}, \rho^{\mathrm{pr}})$ and gone to the current state $(\theta, \rho^\prime)$.

% \begin{figure}
%     \centering
%     \includegraphics[width=\linewidth]{figures/dr_accept_prob.png}
%     \caption{hi}
%     \label{fig:dr_accept}
% \end{figure}

The number of ghost states to be evaluated grows exponentially as $\mathcal{O}(2^k)$ for the $k$th proposal. However these delayed proposals to higher orders are (a) typically made only in the difficult regions of the phase space, and (b) proportional to the number of leapfrog steps required by to traverse the region of high curvature.
Furthermore, these delayed proposals minimize the number of rejections which otherwise cause G-HMC trajectory to reverse due to the flipped momentum. As a result, the extra computational cost of delayed proposals is a favorable tradeoff and keeps DR-G-HMC  computationally competitive with auto-tuned samplers like  NUTS, even when step size reductions are not necessary.

If leapfrog step sizes $\epsilon_1 \ldots \epsilon_K$ are monotonically decreasing, then each subsequent proposal simulates more accurate Hamiltonian dynamics with a higher acceptance probability. This choice is \textit{especially} well suited to multiscale densities because it allows for dynamic, per-iteration step size selection: large step sizes are used in the wide, mouth-like regions while smaller step sizes are used in narrow, neck-like regions. This can be seen in \autoref{fig:funnel_drghmc_stepsize}. While the step size for each proposal can be chosen independently, in this work we choose a geometric progression for step size by reducing them with a factor $r$ at each stage: $\epsilon_k = \epsilon / r^{k-1}$, thereby reducing the number of tuning parameters.

Finally, unlike DR-HMC which retries subsequent proposals after rejecting a state $n_k$ leapfrog steps away from the current state, DR-G-HMC retries proposals after \textit{every} leapfrog step. This increases DR-G-HMC efficiency by (1) selecting  step sizes suited to the local curvature at \emph{every} point in a trajectory and (2) avoiding throwing away long trajectories upon rejection.

\section{EXPERIMENTS}

\begin{figure*}
    \centering
    \includegraphics[scale=0.075]{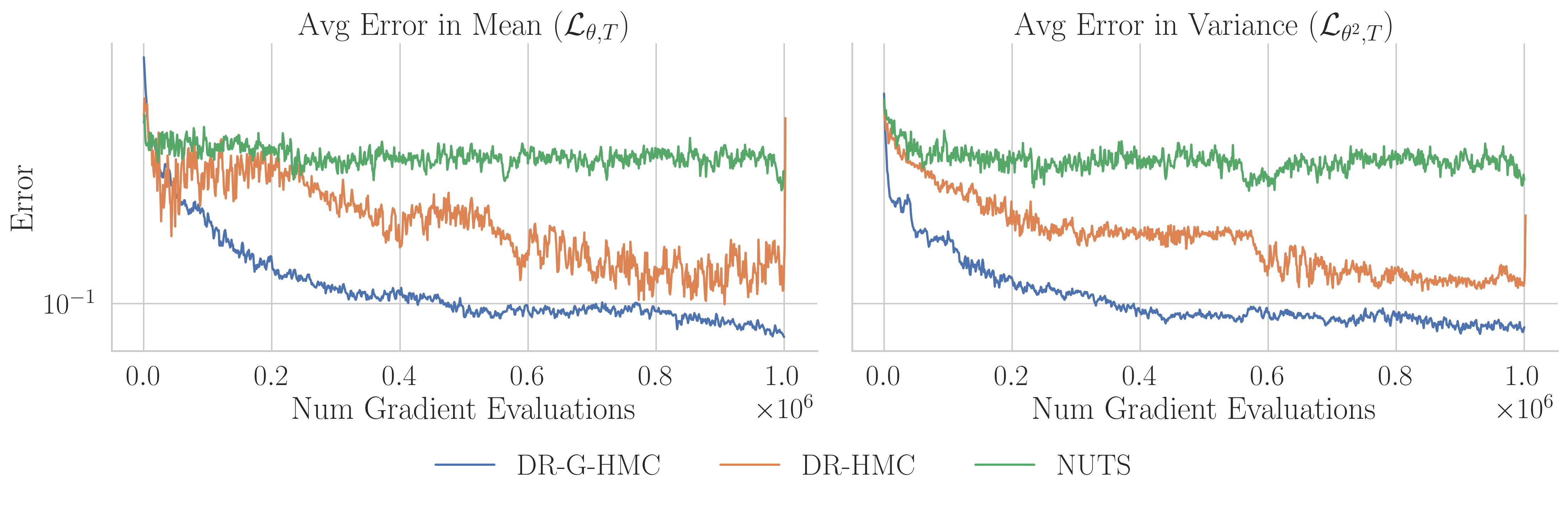}
    \caption{\textbf{Average error vs gradient evaluations for Neal's funnel.} Error in mean ($\mathcal{L}_{\theta, T}$) and variance ($\mathcal{L}_{\theta^2, T}$) averaged over $100$ chains of sampler draws from $10D$ Neal's funnel.  NUTS's error plateaus while delayed rejection methods do not. Our method DR-G-HMC achieves the lowest error.}
    \label{fig:funnel_error_vs_grad}
\end{figure*}

\begin{figure*}
    \centering
    \includegraphics[scale=0.075]{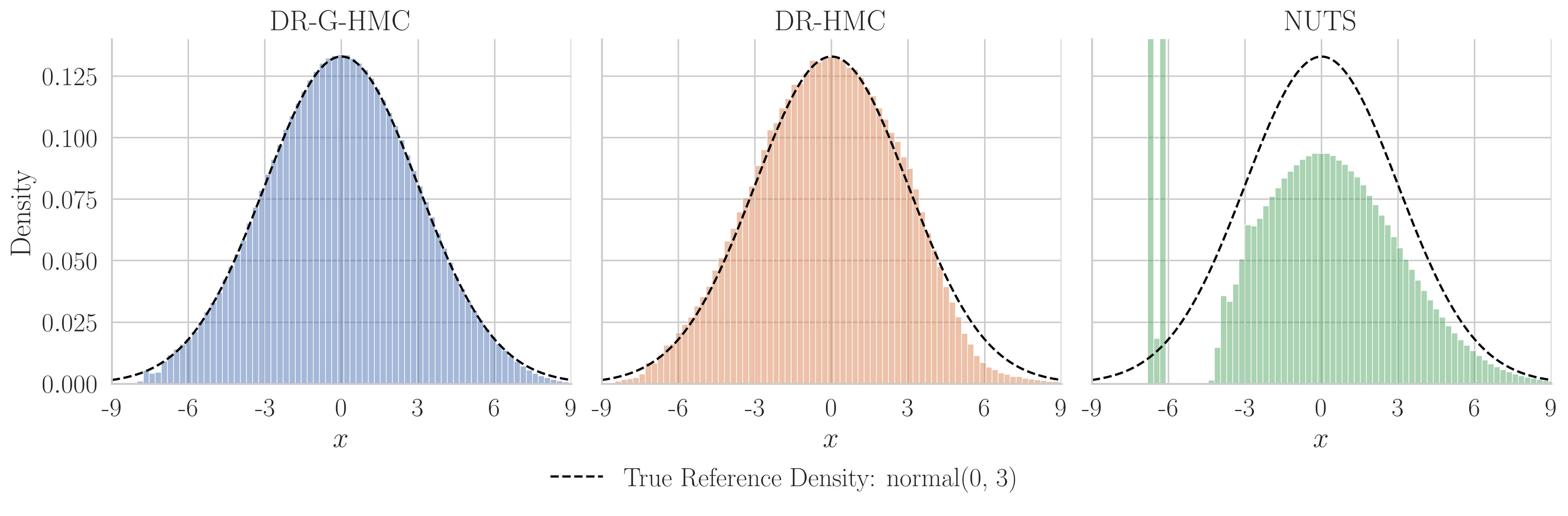}
    \caption{\textbf{Histogram of log scale parameter $x$ in Neal's funnel.} Sampler draws are aggregated across all $100$ chains of $10D$ Neal's funnel. Reference draws are sampled from the known density as $x \sim \textrm{normal}(0,3)$. DR-G-HMC and DR-HMC can sample deep into the highly curved neck ($x \ll 0$) and DR-G-HMC can sample deep into the mouth ($x \gg 0$) with dynamic step size selection, while  NUTS cannot.} 
    \label{fig:funnel_hist}
\end{figure*}

In this section we compare the performance of DR-G-HMC, DR-HMC, and NUTS samplers. Experimental code can be found at \href{https://github.com/gil2rok/drghmc}{https://github.com/gil2rok/drghmc}.

\subsection{Setup}\label{sec:setup}
\vspace*{-12pt}
\paragraph{Measuring Sampler Performance} We evaluate a sampler's performance by measuring absolute standardized error for a fixed computational budget. Since computational cost is dominated by gradient evaluations, we run all samplers with a budget of $T=10^6$ gradient evaluations. We do not use effective sample size (ESS) to measure relaxation because when a chain does not fully explore a challenging target density (e.g. in multiscale densities), the resulting bias is not captured by ESS. In contrast, the absolute standardized error between a sampler's draws and reference draws will detect this bias. For synthetic targets like Neal's funnel, we can generate reference samples independently from the target density. For other problems, we generates reference samples by running  NUTS with target acceptance rate of $0.95$, generating $10^6$ samples, and thinning until roughly independent (as measured by ESS).

Given draws $\theta^{(1)}, \ldots, \theta^{(N)} \in \mathbb{R}^D$, define a distribution 
$$\textstyle \textrm{U}(\theta \mid \theta_1, \ldots, \theta_N) = \sum_{n=1}^N \frac{1}{N} \, \textrm{I}(\theta = \theta_n).$$  
For a function $f:\mathbb{R}^D \rightarrow \mathbb{R}$, expectations are given by
$$\textstyle \mathbb{E}_{\textrm{U}(\theta \mid \theta_1, \ldots, \theta_N)}[f(\theta)] = \frac{1}{N} \sum_{n = 1}^N f(\theta_n).$$  
Given reference draws $\theta^{\textrm{ref}(1)}, \ldots, \theta^{\textrm{ref}(N)}$ and sample draws $\theta^{\textrm{test}(1)}, \ldots, \theta^{\textrm{test}(M)}$ to test generated with a budget of $t$ gradient evaluations, let $p^{\textrm{ref}}(\theta) = \textrm{uniform}(\theta \mid \theta^{\textrm{ref}(1)}, \ldots, \theta^{\textrm{ref}(N)})$ and $p^{\textrm{test}}$ be defined similarly.  For a given expectation function $f$, we define our reference answer as $\mathbb{E}_{p^{\textrm{ref}}}[f(\theta)]$ and our test answer as $\mathbb{E}_{p^{\textrm{test}}}[f(\theta)]$ and estimate the \textit{error} of the test draws as
$\mathbb{E}_{p^{\textrm{test}}}[f(\theta)] - \mathbb{E}_{p^{\textrm{ref}}}[f(\theta)]$.

To report a single statistic for every problem, we standardize the errors to the same scale by transforming to absolute Z score, and take the maximum across all dimensions to identify the worst performing parameter,
\begin{equation}
\mathcal{L}_{f,t}(\theta^{\textrm{test}}, \theta^{\textrm{ref}}) = \max_d 
\dfrac{\left| \mathbb{E}_{p^{\textrm{test}}}[f(\theta)] - \mathbb{E}_{p^{\textrm{ref}}}[f(\theta)]\right|}
      {\mathrm{sd}_{p^{\textrm{test}}}[f(\theta)]},
\end{equation}
where $\text{sd}$ is standard deviation.  We compute the absolute standardized error $\mathcal{L}_{f,t}$ for functions $f(\theta)=\theta_d$ and $f(\theta)=\theta_d^2$ to measure the error in the mean and variance estimates of each parameter.

\paragraph{Choice of Parameters} We generate parameter draws from  NUTS, DR-HMC, and DR-G-HMC samplers, each with $100$ chains. All chains are initialized from a randomly chosen sample from the reference samples, and the corresponding chains in all samplers are initialized from the same sample. Momentum is randomly initialized by sampling from $\text{normal}(0, M)$.

We run  NUTS with the probabilistic programming language Stan (\cite{carpenter2017stan}) using the default target acceptance of $0.80$, and a diagonal mass matrix $M$ to be adapted for all problems except Neal's funnel, where we found the identity matrix to perform better. We implement DR-HMC and DR-G-HMC ourselves and access Stan models with BridgeStan (\cite{Roualdes2023}). All experiments are run in parallel on $128$ CPU cores with $256$GB of memory.

Since DR-HMC and DR-G-HMC adapt step size dynamically as necessary, we run these samplers with an initial step size a factor $c$ larger than the  NUTS adapted step size, $\epsilon = c \epsilon_{\text{ NUTS}}$, for $c=2$. Both algorithms use the same mass-matrix as  NUTS to ensure fair comparison between algorithms. We keep the maximum proposals of $K=3$ and reduction factor of the step size, $\epsilon_k = \epsilon_1 / r^{k-1}$, to be $r=4$.
 
Unlike  NUTS, these algorithms do not adapt trajectory length automatically and hence this parameter needs to be chosen. DR-HMC uses the initial number of steps $n$ as the $90$th percentile of number of leapfrog steps from  NUTS (following \textcite{wu2018faster}).\footnote{Unlike \textcite{wu2018faster}, we do not jitter step size.} For the $k$th proposal, the number of steps are given by  and $n_k = \tau / \epsilon_k $ such that the trajectory length $\tau = \epsilon_1 n$ is maintained constant. For DR-G-HMC, we keep the $n_k = 1$ for all proposals (including the first one) and damping $\gamma = 0.08$. 

\subsection{Sampling From Multiscale Densities}\label{sec:multiscale}

We begin by evaluating the sampler performance on Neal's funnel, a synthetic multiscale density introduced in \textcite{neal2003slice}. A $D$-dimensional funnel defines the density with parameters $x \in \mathbb{R}, y \in \mathbb{R}^{D-1}$, as
$$
    \pi(x, y) = \textrm{normal}(x \mid 0, 3) \prod_{i=1}^{D-1} \textrm{normal}(y_i \mid 0, \exp(x / 2)).
$$
The funnel is extremely difficult to sample because its curvature, scale, and condition number vary dramatically throughout the density. This arises from the log scale parameter $x$ that determines the spread of latent parameters $y_i \sim \textrm{normal}(0, \exp(x/2))$.  Values of $x < 0$ lead to high curvature in $y$ (the neck) while values of $x > 0$ lead to very low curvature in $y$ (the mouth).  Both regions are poorly conditioned and it is a challenge to sample \textit{both} with a fixed step size.

\begin{table*}
\centering
\begin{tabular}{c c c}
\toprule
    % & \multicolumn{2}{c}{Posterior Characteristics} \\
    % \cmidrule(lr){2-3}
    \textbf{Target Posterior} & \textbf{Description} & \textbf{Dimension} \\
    \midrule
    eight schools & Centered parameterization, hierarchical model & 9 \\
    normal100 & Correlated dense covariance & 100 \\
    irt 2pl & Item response theory with additive non-identifiability & 144 \\
    stochastic volatility & High dimensional, correlated, and varying curvature & 503 \\
\bottomrule \\
\end{tabular}
\caption{\textbf{Summary of posterior densities.}}
\label{tab:posteriors}
\end{table*}

\begin{figure*}
    \centering
    \includegraphics[width=0.6\linewidth, height=5.5cm]{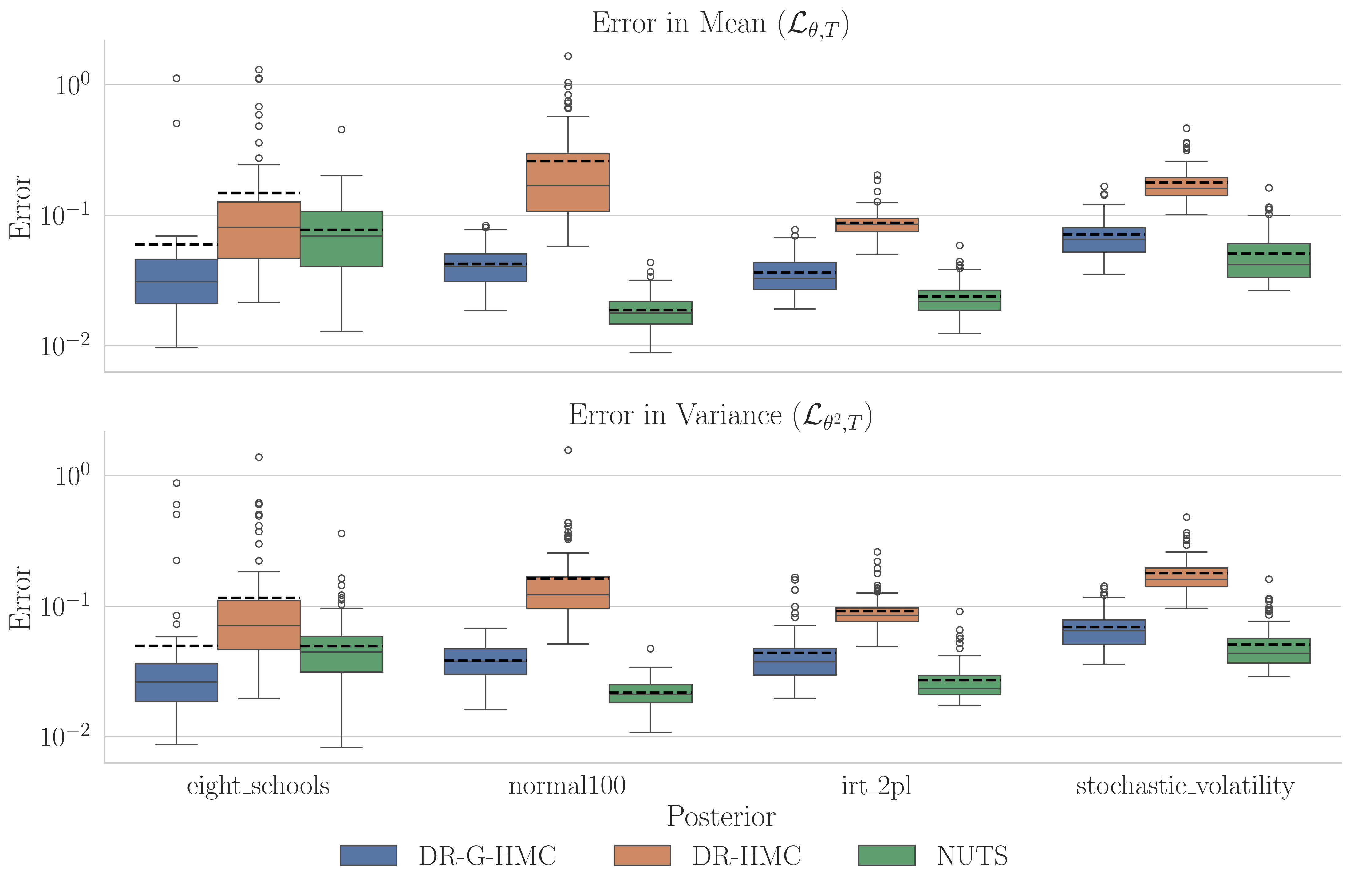}
    \caption{\textbf{DR-G-HMC overcomes the inefficiencies of generalized HMC.} Error in mean ($\mathcal{L}_{\theta,T}$) and variance ($\mathcal{L}_{\theta^2,T}$) is shown on the log scale for $100$ chains. Visual elements represent the following: dashed black line is the mean, solid gray line is the median, colored box is the $(25, 75)$th percentile, whiskers are $1.5$ times the inter-quartile range, and bubbles are outliers.}
    \label{fig:all_errors}
\end{figure*}

In \autoref{fig:funnel_error_vs_grad} we examine the error across gradient evaluations for both the mean ($\mathcal{L}_{\theta, T}$) and variance ($\mathcal{L}_{\theta^2, T}$).  NUTS cannot sample the funnel: it gets the wrong answer as its error plateaus quickly and does \textit{not} decrease with hundreds of thousands of more iterations.  In contrast, delayed rejection methods reduce the error in the mean and variance \textit{without} plateauing, with DR-G-HMC performing best. DR-G-HMC benefits from using a small step size only where necessary and not along entire HMC trajectories. In \autoref{fig:funnel_hist} we perform further investigation by plotting the histogram of the (analytically available) log scale parameter $x \sim \textrm{normal}(0,3)$ from our samplers. This confirms our diagnosis: the delayed rejection methods can sample deep into mouth ($x \gg 0$) and neck ($x \ll 0$) of the funnel, while  NUTS is unable to sample past the $x=-5$ region with a fixed step size. DR-G-HMC and DR-HMC can robustly sample from Neal's funnel while  NUTS will fail due to its multiscale geometry. Similar results are shown in \autoref{app:funnel} for Neal's funnel in $50$, $100$, and $250$ dimensions.

\subsection{Resolving the Inefficiencies of G-HMC}\label{sec:ineffencies} We examine a variety of real-world posterior densities summarized in \autoref{tab:posteriors} and detailed in \autoref{app:posterior_exp}. We seek to demonstrate that in this setting DR-G-HMC alleviates the backtracking inefficiencies of vanilla G-HMC, making it competitive with NUTS.

% \textcite{neal2020nonreversibly} reports that his non-reversible G-HMC scheme is more efficient than HMC, which is in turn more efficient than plain G-HMC. For a more challenging baseline, we use the auto-tuned  NUTS sampler, which outperforms basic HMC (\cite{hoffman2014no}).

We compare DR-G-HMC, DR-HMC, and NUTS samplers as they are of primary interest. However in \autoref{app:samplers} we benchmark a larger variety of samplers including those with and without delayed rejection, momentum refreshment, and multiple leapfrog steps.

 We compute the error in mean ($\mathcal{L}_{\theta, T}$) and variance ($\mathcal{L}_{\theta^2, T}$) in \autoref{fig:all_errors}. DR-G-HMC achieves errors comparable to NUTS across a variety of posteriors while DR-HMC performs notably worse. DR-G-HMC especially outperforms NUTS on the eight schools posterior. Crucially, this demonstrates that delayed rejection removes the inefficiency due to reversing course that plagues G-HMC. This is especially impressive because NUTS extensively \textit{auto-tunes} its parameters while DR-G-HMC does not. Furthermore, since DR-G-HMC outperforms DR-HMC on all posteriors by a significant margin, DR-G-HMC achieves substantial benefit from using a small step size only where necessary.

We emphasize that the main advantage of DR-G-HMC is that it can efficiently sample from multiscale densities (unlike NUTS) \textit{and} from non-multiscale densities (unlike DR-HMC).

\subsection{Robust to Tuning Parameters}\label{sec:robust}

\begin{figure*}
    \centering
    \includegraphics[width=0.7\linewidth]{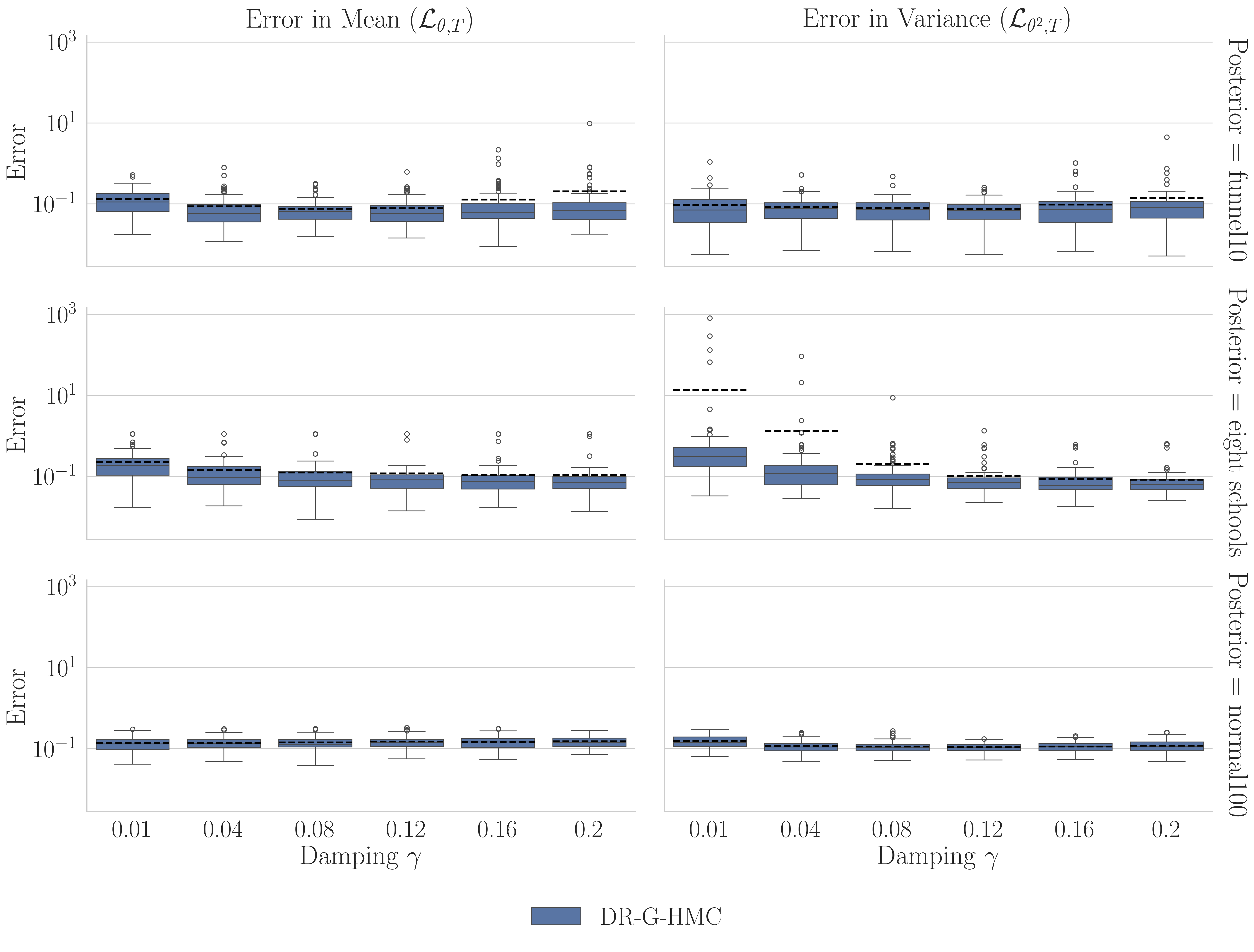}
    \caption{\textbf{DR-G-HMC is robust to the damping tuning parameter $\gamma$.} Error in mean ($\mathcal{L}_{\theta,T}$) and variance ($\mathcal{L}_{\theta^2,T}$) is shown for $100$ chains of various posterior densities. Visual elements represent the following: dashed black line is the mean, solid gray line is the median, colored box is the $(25, 75)$th percentile, whiskers are $1.5$ times the inter-quartile range, and bubbles are outliers.}
    \label{fig:robust_damping}
\end{figure*}

In this section we asses the sensitivity of DR-G-HMC to tuning parameters, especially the damping $\gamma$ that determines directed motion for G-HMC samplers. HMC and its variants are notoriously sensitive to its trajectory length $\tau = n \epsilon$ that determines directed motion for HMC samplers. Before  NUTS auto-tuned this parameter, HMC was incredibly difficult to use.

In \autoref{fig:robust_damping} we run DR-G-HMC on multiple posteriors for $T=10^5$ gradient evaluations with damping values $\gamma = [0.01, 0.04, 0.08, 0.12, 0.16, 0.2]$. We measure error in mean ($\mathcal{L}_{\theta, T}$) and variance ($\mathcal{L}_{\theta^2, T}$) and find that DR-G-HMC errors are remarkably stable across a wide range of damping values and posteriors. This indicates DR-G-HMC is insensitive to the damping value $\gamma$ and perhaps has an easier time making directed motion than HMC. Furthermore, in our box plot, we find the variance across DR-G-HMC chains is often larger than the variance between different damping values. This suggests chain initialization plays a larger role than the damping itself. \autoref{app:robust} demonstrates robustness to additional DR-G-HMC tuning parameters. \autoref{app:runtime} compares wallclock runtime of DR-G-HMC and DR-HMC.

\section{DISCUSSION}

We have not addressed the open problem of automatically tuning the DR-G-HMC parameters. The damping factor $\gamma$ is especially important because it determines how much momentum is conserved, analogous to trajectory length in HMC. In prior work, the MEADS algorithm (\cite{MEADs}) auto-tuned the parameters for the non-reversible G-HMC slice sampling scheme of \textcite{neal2003slice}, including the damping factor $\gamma$. If similar work can be done, DR-G-HMC may see a performance increase similar to  NUTS over HMC.

\section{CONCLUSION}

We introduced the DR-G-HMC sampler which applies delayed rejection to G-HMC to adapt to local curvature with dynamic step size selection. By combining the useful features of DR-HMC and G-HMC, our sampler \textit{simultaneously} solves the directed motion problem of G-HMC and provides multiscale density sampling for hierarchical models. DR-G-HMC ultimately outperforms both its individual components, while being competitive with  NUTS and robust to tuning parameters.

\subsubsection*{References}

\printbibliography[heading=none]

@article{Modi_2023,
   title={Delayed rejection {H}amiltonian {M}onte {C}arlo for sampling multiscale distributions},
   volume={1},
   ISSN={1936-0975},
   url={http://dx.doi.org/10.1214/23-BA1360},
   DOI={10.1214/23-ba1360},
   number={1},
   journal={Bayesian Analysis},
   publisher={Institute of Mathematical Statistics},
   author={Modi, Chirag and Barnett, Alex and Carpenter, Bob},
   year={2023},
   month=jan 
}

@misc{neal2020nonreversibly,
      title={Non-reversibly updating a uniform [0,1] value for {M}etropolis accept/reject decisions}, 
      author={Neal, Radford M},
      year={2020},
      eprint={2001.11950},
      archivePrefix={arXiv},
      primaryClass={stat.CO}
}

@InProceedings{autoMALA,
  title = 	 { {autoMALA}: Locally adaptive {M}etropolis-adjusted {L}angevin algorithm },
  author =       {Biron-Lattes, Miguel and Surjanovic, Nikola and Syed, Saifuddin and Campbell, Trevor and Bouchard-Cote, Alexandre},
  booktitle = 	 {Proceedings of The 27th International Conference on Artificial Intelligence and Statistics},
  pages = 	 {4600--4608},
  year = 	 {2024},
  editor = 	 {Dasgupta, Sanjoy and Mandt, Stephan and Li, Yingzhen},
  volume = 	 {238},
  series = 	 {Proceedings of Machine Learning Research},
  month = 	 may,
  publisher =    {PMLR},
  url = 	 {https://proceedings.mlr.press/v238/biron-lattes24a.html}
}

@InProceedings{MEADs,
  title = 	 { Tuning-Free Generalized {H}amiltonian {M}onte {C}arlo},
  author =       {Hoffman, Matthew D and Sountsov, Pavel},
  booktitle = 	 {Proceedings of The 25th International Conference on Artificial Intelligence and Statistics},
  pages = 	 {7799--7813},
  year = 	 {2022},
  editor = 	 {Camps-Valls, Gustau and Ruiz, Francisco J. R. and Valera, Isabel},
  volume = 	 {151},
  series = 	 {Proceedings of Machine Learning Research},
  month = 	 mar,
  publisher =    {PMLR},
  url = 	 {https://proceedings.mlr.press/v151/hoffman22a.html}
}

@article{horowitz1991generalized,
  title={A generalized guided {M}onte {C}arlo algorithm},
  author={Horowitz, Alan M},
  journal={Physics Letters B},
  volume={268},
  number={2},
  pages={247--252},
  year={1991},
  publisher={Elsevier}
}

@article{green2001delayed,
  title={Delayed rejection in reversible jump {M}etropolis--{H}astings},
  author={Green, Peter J and Mira, Antonietta},
  journal={Biometrika},
  volume={88},
  number={4},
  pages={1035--1053},
  year={2001},
  publisher={Biometrika Trust}
}

@article{metropolis1953equation,
  title={Equation of state calculations by fast computing machines},
  author={Metropolis, Nicholas and Rosenbluth, Arianna W and Rosenbluth, Marshall N and Teller, Augusta H and Teller, Edward},
  journal={The journal of chemical physics},
  volume={21},
  number={6},
  pages={1087--1092},
  year={1953},
  publisher={American Institute of Physics}
}

@article{hastings1970monte,
  title={Monte {C}arlo sampling methods using {M}arkov chains and their applications},
  author={Hastings, W Keith},
  journal={Biometrika},
  volume={57},
  number={1},
  pages={97--97},
  year={1970}
}

@article{bedard2014scaling,
  title={Scaling analysis of delayed rejection {MCMC} methods},
  author={B{\'e}dard, Myl{\`e}ne and Douc, Randal and Moulines, Eric},
  journal={Methodology and Computing in Applied Probability},
  volume={16},
  number={4},
  pages={811--838},
  year={2014},
  publisher={Springer}
}

@inproceedings{sohl2014hamiltonian,
  title={Hamiltonian {M}onte {C}arlo without detailed balance},
  author={Sohl-Dickstein, Jascha and Mudigonda, Mayur and DeWeese, Michael},
  booktitle={International Conference on Machine Learning},
  pages={719--726},
  year={2014},
  organization={PMLR}
}

@article{campos2015extra,
  title={Extra chance generalized hybrid {M}onte {C}arlo},
  author={Campos, C{\'e}dric M and Sanz-Serna, Jes{\'u}s Mar{\'\i}a},
  journal={Journal of Computational Physics},
  volume={281},
  pages={365--374},
  year={2015},
  publisher={Elsevier}
}

@article{duane1987hybrid,
  title={Hybrid {m}onte {c}arlo},
  author={Duane, Simon and Kennedy, Anthony D and Pendleton, Brian J and Roweth, Duncan},
  journal={Physics letters B},
  volume={195},
  number={2},
  pages={216--222},
  year={1987},
  publisher={Elsevier}
}

@inbook{neal2011mcmc,
  title={{MCMC} using {H}amiltonian dynamics},
  author={Neal, Radford M},
  booktitle={Handbook of Markov Chain Monte Carlo},
  editor={Brooks, Steve and Gelman, Andrew and Jones, Galin L and Meng, Xiao-Li},
  year={2011},
  publisher={Chapman and Hall/CRC}
}

@article{betancourt2017conceptual,
  title={A conceptual introduction to {H}amiltonian {M}onte {C}arlo},
  author={Betancourt, Michael},
  journal={arXiv preprint arXiv:1701.02434},
  year={2017}
}

@article{hoffman2014no,
  title={The {N}o-{U}-{T}urn sampler: adaptively setting path lengths in {H}amiltonian {M}onte {C}arlo.},
  author={Hoffman, Matthew D and Gelman, Andrew},
  journal={J. Mach. Learn. Res.},
  volume={15},
  number={1},
  pages={1593--1623},
  year={2014}
}

@article{carpenter2017stan,
  title={Stan: A probabilistic programming language},
  author={Carpenter, Bob and Gelman, Andrew and Hoffman, Matthew D and Lee, Daniel and Goodrich, Ben and Betancourt, Michael and Brubaker, Marcus A and Guo, Jiqiang and Li, Peter and Riddell, Allen},
  journal={Journal of statistical software},
  volume={76},
  year={2017},
  publisher={NIH Public Access}
}

@article{atchade2006adaptive,
  title={An adaptive version for the {M}etropolis adjusted {L}angevin algorithm with a truncated drift},
  author={Atchad{\'e}, Yves F},
  journal={Methodology and Computing in applied Probability},
  volume={8},
  pages={235--254},
  year={2006},
  publisher={Springer}
}

@article{marshall2012adaptive,
  title={An adaptive approach to {L}angevin {MCMC}},
  author={Marshall, Tristan and Roberts, Gareth},
  journal={Statistics and Computing},
  volume={22},
  pages={1041--1057},
  year={2012},
  publisher={Springer}
}

@article{kleppe2016adaptive,
  title={{A}daptive Step Size Selection for {H}essian-Based Manifold {L}angevin Samplers},
  author={Kleppe, Tore Selland},
  journal={Scandinavian Journal of Statistics},
  volume={43},
  number={3},
  pages={788--805},
  year={2016},
  publisher={Wiley Online Library}
}

@article{coullon2023efficient,
  title={{E}fficient and generalizable tuning strategies for stochastic gradient {MCMC}},
  author={Coullon, Jeremie and South, Leah and Nemeth, Christopher},
  journal={Statistics and Computing},
  volume={33},
  number={3},
  pages={66},
  year={2023},
  publisher={Springer}
}

@inproceedings{campbell2021gradient,
  title={A gradient based strategy for {H}amiltonian {M}onte {C}arlo hyperparameter optimization},
  author={Campbell, Andrew and Chen, Wenlong and Stimper, Vincent and Hernandez-Lobato, Jose Miguel and Zhang, Yichuan},
  booktitle={International Conference on Machine Learning},
  pages={1238--1248},
  year={2021},
  organization={PMLR}
}

@article{girolami2011riemann,
  title={{R}iemann manifold {L}angevin and {H}amiltonian {M}onte {C}arlo methods},
  author={Girolami, Mark and Calderhead, Ben},
  journal={Journal of the Royal Statistical Society Series B: Statistical Methodology},
  volume={73},
  number={2},
  pages={123--214},
  year={2011},
  publisher={Oxford University Press}
}

@article{pourzanjani2019implicit,
  title={{I}mplicit {H}amiltonian {M}onte {C}arlo for sampling multiscale distributions},
  author={Pourzanjani, Arya A and Petzold, Linda R},
  journal={arXiv preprint arXiv:1911.05754},
  year={2019}
}

@inproceedings{brofos2021evaluating,
  title={{E}valuating the implicit midpoint integrator for {R}iemannian {H}amiltonian {M}onte {C}arlo},
  author={Brofos, James and Lederman, Roy R},
  booktitle={International Conference on Machine Learning},
  pages={1072--1081},
  year={2021},
  organization={PMLR}
}

@article{betancourt2015hamiltonian,
  title={{H}amiltonian {M}onte {C}arlo for hierarchical models},
  author={Betancourt, Michael and Girolami, Mark},
  journal={Current trends in Bayesian methodology with applications},
  volume={79},
  number={30},
  pages={2--4},
  year={2015},
  publisher={CRC Press Boca Raton, FL}
}

@article{haario2006dram,
  title={{DRAM}: efficient adaptive {MCMC}},
  author={Haario, Heikki and Laine, Marko and Mira, Antonietta and Saksman, Eero},
  journal={Statistics and computing},
  volume={16},
  pages={339--354},
  year={2006},
  publisher={Springer}
}

@article{tierney1999some,
  title={{S}ome adaptive {M}onte {C}arlo methods for {B}ayesian inference},
  author={Tierney, Luke and Mira, Antonietta},
  journal={Statistics in medicine},
  volume={18},
  number={17-18},
  pages={2507--2515},
  year={1999},
  publisher={Wiley Online Library}
}

@article{gelman1997weak,
  title={{W}eak convergence and optimal scaling of random walk {M}etropolis algorithms},
  author={Gelman, Andrew and Gilks, Walter R and Roberts, Gareth O},
  journal={The annals of applied probability},
  volume={7},
  number={1},
  pages={110--120},
  year={1997},
  publisher={Institute of Mathematical Statistics}
}

@article{neal2003slice,
  title={{S}lice sampling},
  author={Neal, Radford M},
  journal={The annals of statistics},
  volume={31},
  number={3},
  pages={705--767},
  year={2003},
  publisher={Institute of Mathematical Statistics}
}

@article{wu2018faster,
  title={Faster {H}amiltonian {M}onte {C}arlo by learning leapfrog scale},
  author={Wu, Changye and Stoehr, Julien and Robert, Christian P},
  journal={arXiv preprint arXiv:1810.04449},
  year={2018}
}

@article{rubin1981estimation,
  title={Estimation in parallel randomized experiments},
  author={Rubin, Donald B},
  journal={Journal of Educational Statistics},
  volume={6},
  number={4},
  pages={377--401},
  year={1981},
  publisher={Sage Publications Sage CA: Thousand Oaks, CA}
}

@article{gelman2021bayesian,
  title={Bayesian Data Analysis Third edition},
  author={Gelman, Andrew and Carlin, John B and Stern, Hal S and Dunson, David B and Vehtari, Aki and Rubin, Donald B},
  journal={Issue: April},
  year={2021}
}

@article{kim1998stochastic,
  title={Stochastic volatility: likelihood inference and comparison with ARCH models},
  author={Kim, Sangjoon and Shephard, Neil and Chib, Siddhartha},
  journal={The review of economic studies},
  volume={65},
  number={3},
  pages={361--393},
  year={1998},
  publisher={Wiley-Blackwell}
}

@book{gelman2006data,
  title={Data analysis using regression and multilevel/hierarchical models},
  author={Gelman, Andrew and Hill, Jennifer},
  year={2006},
  publisher={Cambridge university press}
}

@article{Roualdes2023, doi = {10.21105/joss.05236}, url = {https://doi.org/10.21105/joss.05236}, year = {2023}, publisher = {The Open Journal}, volume = {8}, number = {87}, pages = {5236}, author = {Edward A. Roualdes and Brian Ward and Bob Carpenter and Adrian Seyboldt and Seth D. Axen}, title = {BridgeStan: Efficient in-memory access to the methods of a Stan model}, journal = {Journal of Open Source Software} }

%%%%%%%%%%%%%%%%%%%%%%%%%%%%%%%%%%%%%%%%%%%%%%%%%%%%%%%%%%%%
\section*{Checklist}

 \begin{enumerate}

 \item For all models and algorithms presented, check if you include:
 \begin{enumerate}
   \item A clear description of the mathematical setting, assumptions, algorithm, and/or model. [Yes] The DR-G-HMC algorithm is detailed in \autoref{alg:DR-G-HMC}.
   \item An analysis of the properties and complexity (time, space, sample size) of any algorithm. [Yes] In DR-G-HMC, the number of model + gradient evaluations for the $k$th proposal attempt is $\mathcal{O}(2^k)$ detailed in \autoref{sec:drghmc}, along with implications of caching. 
   \item (Optional) Anonymized source code, with specification of all dependencies, including external libraries. [Yes] Annonymized source is submitted.
 \end{enumerate}

 \item For any theoretical claim, check if you include:
 \begin{enumerate}
   \item Statements of the full set of assumptions of all theoretical results. [Yes] Limited assumptions are used in the paper, but they can be found in \autoref{app:drghmc}.
   \item Complete proofs of all theoretical results. [Yes] \autoref{app:accept_prob} for derivation of acceptance probability $\alpha$ and \autoref{app:invariant} for proof of detailed balance.
   \item Clear explanations of any assumptions. [Yes] See above.
 \end{enumerate}

 \item For all figures and tables that present empirical results, check if you include:
 \begin{enumerate}
   \item The code, data, and instructions needed to reproduce the main experimental results (either in the supplemental material or as a URL). [Yes] Included in the annonymized source code.
   \item All the training details (e.g., data splits, hyperparameters, how they were chosen). [Yes] Detailed in \autoref{sec:setup}.
     \item A clear definition of the specific measure or statistics and error bars (e.g., with respect to the random seed after running experiments multiple times). [Yes] Detailed in \autoref{sec:setup}.
     \item A description of the computing infrastructure used. (e.g., type of GPUs, internal cluster, or cloud provider). [Yes] Detailed in \autoref{sec:setup}.
 \end{enumerate}

 \item If you are using existing assets (e.g., code, data, models) or curating/releasing new assets, check if you include:
 \begin{enumerate}
   \item Citations of the creator If your work uses existing assets. [Yes] For major software used, I cite them, primarily \textcite{carpenter2017stan, Roualdes2023}.
   \item The license information of the assets, if applicable. [Yes] We checked relevant licenses for software used.
   \item New assets either in the supplemental material or as a URL, if applicable. [Yes] We include annonymized source code upon submission.
   \item Information about consent from data providers/curators. [Not Applicable] Not using any public datasests.
   \item Discussion of sensible content if applicable, e.g., personally identifiable information or offensive content. [Not Applicable] No sensitive content used.
 \end{enumerate}

 \item If you used crowdsourcing or conducted research with human subjects, check if you include:
 \begin{enumerate}
   \item The full text of instructions given to participants and screenshots. [Not Applicable] Research does not involve humans.
   \item Descriptions of potential participant risks, with links to Institutional Review Board (IRB) approvals if applicable. [Not Applicable] Research does not involve humans.
   \item The estimated hourly wage paid to participants and the total amount spent on participant compensation. [Not Applicable] Research does not involve humans.
 \end{enumerate}

 \end{enumerate}

\clearpage
\onecolumn
\appendix

\section{DR-G-HMC TECHNICAL DETAILS} \label{app:drghmc}

First we review technical background on Markov chain Monte Carlo, Metropolis-Hastings, and HMC in \autoref{sec:MH}. Then we prove that DR-G-HMC maintains an invariance over the target density in \autoref{app:invariant}. Finally we derive the acceptance probability of DR-G-HMC that maintains detailed balance in \autoref{app:accept_prob}. We use notation that differs from the remainder of the paper in this appendix.

\subsection{Background}\label{sec:MH}

\subsubsection{Markov Chain Monte Carlo} Markov chain Monte Carlo (MCMC) methods seek to generate samples $x \in S$ from a probability density function (pdf) $\pi$ that is absolutely continuous (AC) over some state space $S$. MCMC methods do so by constructing a Markov chain with $\pi$ as its stationary distribution using a transition kernel $k(x,y)$ that gives the pdf of transitioning from the current state $x$ to the next state $y$.

\begin{definition}[Transition kernel normalization]
    The transition kernel $k$ is normalized as
    \begin{equation}\label{eq:normalization}
        \int k(x,y) dy = 1 \qquad \forall x \in S.
    \end{equation}
\end{definition}

\begin{definition}[Invariance]
    A Markov chain with stationary distribution $\pi$ must be invariant to the transition kernel $k$ as
    \begin{equation}\label{eq:invariance}
        \int \pi(x) k(x,y) dx = \pi(y) \qquad \forall y \in S.
    \end{equation}
\end{definition}

\begin{definition}[Detailed balance]\label{def:detailed_balance} An AC transition kernel $k$ satisfies detailed balance if
    \begin{equation}\label{eq:detailed_balance}
        \pi(x) k(x,y) = \pi(y) k(y,x).
    \end{equation}
    A non-AC transition kernel $k$ maintains detailed balance if
    \begin{equation}\label{eq:detailed_balance_weak}
        \int_A \int_B \pi(x) k(x, y) dx dy = \int_A \int_B \pi(y) k(y,x) dx dy \qquad \text{for all subsets } A, B \subset S
    \end{equation}
    since the two sides of \autoref{eq:detailed_balance} may not always be defined.
\end{definition}

\begin{proposition}
    If an AC transition kernel maintains detailed balance, it satisfies $\pi$-invariance by substituting \autoref{eq:detailed_balance} into \autoref{eq:invariance} and applying \autoref{eq:normalization}. (Non-AC transition kernels also maintain $\pi$-invariance but the more complex proof is omitted here.)
\end{proposition}

\subsubsection{Metropolis-Hastings} The Metropolis-Hastings algorithm (\cite{metropolis1953equation}; \cite{hastings1970monte}) defines an AC proposal kernel $q(x,y)$ that gives the pdf of proposing a new state $y$ from the current state $x$. With some probability $\alpha(x,y)$ the proposal made with $q(x,y)$ becomes the next state of the Markov chain. Otherwise, the Markov chain remains at the current state $x$. The Metropolis-Hastings transition kernel is
\begin{equation*}
    k(x,y) = q(x,y) \alpha(x,y) + \delta_x(y) r(x)
\end{equation*}
in which one can transition to state $y$ by (1) accepting a proposal to move from $x$ to $y$ (2) by rejecting a proposal to leave $y$. To show Metropolis-Hastings is $\pi$-invariant, we must pick the acceptance probability $\alpha$ that maintains detailed balance (\autoref{def:detailed_balance}) as
\begin{equation*}
    \pi(x) \left( q(x,y) \alpha(x,y) + \delta_x(y) r(x) \right) = \pi(y) \left( q(y, x) \alpha(y, x) + \delta_y(x) r(y) \right)
\end{equation*}
for the AC transition kernel $q$. Because the second term of the transition kernel, $\delta_x(y)r(x)$, is $x \leftrightarrow y$ symmetric, we need only consider the first term and simplify to
\begin{equation*}
    \pi(x) q(x,y) \alpha(x,y) = \pi(y) q(y, x) \alpha(y, x).
\end{equation*}
The most efficient acceptance probability is then $\alpha(x,y) = \min \left(1, \frac{\pi(y) q(y, x)}{\pi(x) q(x, y)} \right)$ because for every $x,y \in S$, either $\alpha(x,y)$ or $\alpha(y,x)$ is equal to one.

\subsection{DR-G-HMC is \texorpdfstring{$\pi$}{pi}-Invariant}\label{app:invariant}

In this section we prove that DR-G-HMC updates leave the target density $\pi$-invariant over the state space $S=\mathbb{R}^{2D}$.

\begin{definition}[Volume-Preserving] A map $F: \mathbb{R}^{2D} \rightarrow \mathbb{R}^{2D}$ is volume preserving if 
\begin{equation*}
    \int_B dx = \int_{F(B)} dx \qquad \text{for all subsets }B \subset \mathbb{R}^{2D}
\end{equation*}
where $F(B) := \{ F(x) : x \in B \}$ denotes the image of the set $B$.
\end{definition}

\begin{definition}[Involution]
    A map $F: \mathbb{R}^{2D} \rightarrow \mathbb{R}^{2D}$ is an involution if $F^{-1} = F$ as maps or equivalently $F^2 = I$ for identity map $I$.
\end{definition}

\begin{definition}[Shear]
    Any map on $\mathbb{R}^{2D}$ of the form $(\theta, \rho) \rightarrow (\theta + G(\rho), \rho)$ or $(\theta, \rho) \rightarrow (\theta, \rho + G(\theta))$ for differentiable map $G: \mathbb{R}^D \rightarrow \mathbb{R}^D$ is called a shear.
\end{definition}

\begin{proposition}\label{prop}
    Any shear is volume-preserving.
\end{proposition}

\begin{proof}
    The Jacobian of some shear $F$ with $D \times D$ blocks is $DF = \begin{bmatrix} I_D & J_G \\ 0 & I_D \end{bmatrix}$ or $DF = \begin{bmatrix} I_D & 0 \\ J_G & I_D \end{bmatrix}$. In both cases $\text{det}(DF) = 1$ implying $F$ is volume-preserving.
\end{proof}

\begin{lemma}\label{lemma:F_vol_pre_inv}
    Consider the momentum flip map $P$ and $n \in \mathbb{N}$ compositions of the leapfrog map $L_\epsilon$ with $\epsilon > 0$ defined in \autoref{alg:DR-G-HMC}. Then the map $F: \mathbb{R}^{2D} \rightarrow \mathbb{R}^{2D}$ defined as $F=P L_\epsilon^n$ is a volume-preserving involution.
\end{lemma}

\begin{proof}
    $P$ is trivially volume-preserving and $L_\epsilon$ is the composition of three shears, each of which is volume-preserving by \autoref{prop}. Thus the composition $F = P L_\epsilon^n$ is volume preserving. $L_\epsilon$ is time-reversible such that if $L_\epsilon(\theta^{(t)}, \rho^{(t)}) = (\theta^{(t+1)}, \rho^{(t+1)})$, then $L_\epsilon(\theta^{(t+1)}, -\rho^{(t+1)})=(\theta^{(t)}, -\rho^{(t)})$ by inspection of the three steps of the leapfrog map. Since this holds for any $n$, we know $P L_\epsilon^n = (P L_\epsilon^n)^{-1}$ and thus $F$ is an involution.
\end{proof}

\begin{lemma}\label{eq:lemma}
    The Metropolis-Hastings algorithm with a deterministic proposal kernel $q_F(x,y) = \delta(y-F(x))$ and acceptance probability $\alpha$ obeying
    \begin{equation}\label{eq:determ_prop_kernel}
        \pi(x) \alpha(x,y) = \pi(y) \alpha(y, x) \qquad \forall x,y \in \mathbb{R}^{2D}
    \end{equation}
    maintains an invariance over a density $\pi$ if the map $F$ is a volume-preserving involution.
\end{lemma}

\begin{proof}
    For Metropolis-Hastings updates to remain $\pi$-invariant, we must maintain detailed balance in the weak sense of \autoref{eq:detailed_balance_weak} as
    \begin{equation*}
        \int_A \int_B \pi(x) \alpha(x, y) q_F(x,y) dx dy = \int_A \int_B \pi(y) \alpha(x, y) q_F(y, x) dx dy
    \end{equation*}
    for all subsets $A, B \in \mathbb{R}^D$. We show this is true by transforming the left hand side to the right hand side in a series of steps that require volume-preservation $F=F^{-1}$ and involution $\text{det} (DF)=1$:

    \begin{align*}
         \int_A \int_B &\pi(x) \alpha(x, y) q_F(x,y) dx dy \\
         &= \int_A \int_B \pi(y) \alpha(y, x) \delta(y - F(x)) dx dy  \qquad  \text{by substitution of \autoref{eq:determ_prop_kernel}} \\
         &= \int_{B \bigcap F^{-1}(A)} \pi(F(x)) \alpha(F(x), x) dx \qquad \text{by shifting property of the delta distribution} \\
         &= \int_{F(B) \bigcap A} \pi(y) \alpha(y, F^{-1}(y)) \cdot \mid \text{det}(DF(F^{-1}(y)) \mid ^{-1} dy  \qquad \text{by substituting } y=F(x) \\
         &= \int_{F^{-1}(B) \bigcup A} \pi(y) \alpha(y, F(y)) dy \qquad \text{by } F^{-1}=F \text{ and unity Jacobian factor} \\
         &= \int_A \int_B \pi(y) \alpha(y, x) \delta(x - F(y)) dx dy \qquad \text{by shifting property of the delta distribution}
    \end{align*}
\end{proof}

\begin{theorem}[DR-G-HMC is invariant over the Gibbs density $\tilde{\pi}$]\label{eq:theorem}
    Let $\pi$ be a continuous, differentiable pdf over $\mathbb{R}^{D}$ with an associated Gibbs density $\tilde{\pi}$ over $\mathbb{R}^{2D}$ in \autoref{eq:gibbs}. The Markov chain with DR-G-HMC updates, given by the composition of a Gibbs step and up to $K$ Metropolis-Hastings step, maintains the Gibbs density $\tilde{\pi}$ as an invariant density.
\end{theorem}

\begin{proof}
    The Gibbs step is $\tilde{\pi}$-invariant because it preserves the conditional over the momentum $\rho$ while leaving the position $\theta$ unaltered. The Metropolis-Hastings steps are $\tilde{\pi}$-invariant by \autoref{eq:lemma} because they use a deterministic proposal kernel with a volume-preserving involution $F = P L_\epsilon^n$.
\end{proof}

\begin{corollary}
    If DR-G-HMC is invariant over the Gibbs density $\tilde{\pi}$, then it is invariant over the target density $\pi$.
\end{corollary}

\subsection{DR-G-HMC Acceptance Probability}\label{app:accept_prob}

DR-G-HMC produces samples $(\theta, \rho)$ from the Gibbs density $\tilde{\pi}$ (\autoref{eq:gibbs}) by generating states from a $\tilde{\pi}$-invariant Markov chain. In the Metropolis-Hastings updates of DR-G-HMC, we make upto $K$ proposals and accept it with some probability that maintains detailed balance (\autoref{def:detailed_balance}). In this section we derive this acceptance probability, displayed in \autoref{eq:accept_prob}. 

We will build towards the acceptance probability of arbitrarily many proposal attempts by first deriving the acceptance probability for $K=1,2,3$. For each proposal attempt, we define the proposal kernel $q$, then the transition kernel $k$, and finally the acceptance probability $\alpha$.

To simplify notation, we represent states from the $\tilde{\pi}$-invariant Markov chain as $x=(\theta, \rho)$. Furthermore, although our transition kernel $k$ will be non-AC, we will write detailed balance statements as \autoref{eq:detailed_balance} with the understanding that they will be interpreted as the weak sense in \autoref{eq:detailed_balance_weak}.

\subsubsection{One Proposal Acceptance Probability}

Let the $\tilde{\pi}$-invariant Markov chain have current state $x$ and proposed state $y$. We define a non-AC, deterministic proposal kernel with the proposal map $F_1 = P L_{\epsilon_1}^{n_1}$ as
\begin{equation*}
    q_1(x,y) = \delta(y-F_1(x)).
\end{equation*}
The proposal kernel $q_1$ puts all its probability mass on $y=F_1(x)$ and zero everywhere else. This defines the transition kernel
\begin{equation*}
    k(x,y) = q_1(x,y) \alpha_1(x,y) + \delta_x(y) r_1(x)
\end{equation*}
with $x \leftrightarrow y$ symmetry in the rejection term $\delta_x(y)r_1(x)$ and the factor $q_1(x,y)$ (by $F$ in \autoref{lemma:F_vol_pre_inv}). With these symmetries, maintaining detailed balance reduces to satisfying
\begin{equation*}
    \tilde{\pi}(x) \alpha_1(x,y) = \tilde{\pi}(y) \alpha_1(y,x).
\end{equation*}
We choose the acceptance probability to maintain detailed balance with the maximum chance of acceptance as
\begin{equation*}
    \alpha_1(x,y) = \min \left( 1, \frac{ \tilde{\pi}(y)}{\tilde{\pi} (x)} \right).
\end{equation*}

\subsubsection{Two Proposals Acceptance Probability} Let the $\tilde{\pi}$-invariant Markov chain have current state $x$, rejected state $s_1$, and proposed state $y$. We define a non-AC, deterministic proposal kernel with the proposal map $F_2 = P L_{\epsilon_2}^{n_2}$ as
\begin{equation*}
    q_2(x, s_1, y) = \delta(y-F_2(x)).
\end{equation*}
This defines a transition kernel that must account for all the ways to get to state $y$: (1) accepting the first proposal $q_1(x,y)$ with probability $\alpha_1(x,y)$; (2) accepting the second proposal $q_2(x, s_1, y)$ with some new probability $\alpha_2(x, y)$; (3) rejecting the second proposal. Cases (2) and (3) must integrate over all possible rejected proposals $s_1$, giving us
\begin{equation*}
    k(x,y) = q_1(x,y) \alpha_1(x, y) + \int q_1(x, s_1) [1 - \alpha_1(x, s_1)] [ q_2(x, s_1, y) \alpha_2(x, y) + r_2(x) \delta_x(y)] ds_1
\end{equation*}
with probability of rejecting the second proposal as $r_2$.
Since the first term already satisfies detailed balance and $ r_2(x) \delta_x(y)$ is $x \leftrightarrow y$ symmetric, maintaining detailed balance reduces to plugging in deterministic proposal kernels $q_1, q_2$ and satisfying 
\begin{align*}
    \int \tilde{\pi}(x) q_1(x, s_1) [1 - \alpha_1(x,s)] q_2(x, s_1, y) \alpha_2(x, s_1, y) ds_1 = \\
    \int \tilde{\pi}(y) q_1(y, s_1^\prime) [1 - \alpha_1(y, s_1^\prime)] q_2(y, s_1^\prime, x) \alpha_2(y, s_1^\prime, x) ds_1^\prime
\end{align*}
for dummy variables $s_1, s_1^\prime$. To pick the acceptance probability $\alpha_2$ that satisfies detailed balance, we plug in deterministic proposal kernels $q_1, q_2$ as
\begin{align*}
    \int \tilde{\pi}(x) \delta(s_1 - F_1(x)) [ 1 - \alpha_1(x,s_1)] \delta(y - F_2(x)) \alpha_2(x, s_1, y) ds_1 = \\
    = \int \tilde{\pi}(y) \delta(s_1^\prime - F_1(y)) [1 - \alpha_1(y, s_1^\prime)] \delta(x - F_2(y)) \alpha_2(y, s_1^\prime, x) d s_1^\prime.
\end{align*}
Unlike the derivation in \textcite{tierney1999some}, we \textit{evaluate} this integral as
\begin{align*}
    \tilde{\pi}(x) [1 - \alpha_1(x, F_1(x))] \delta(y - F_2(x)) \alpha_2(x, F_1(x), y) = \\
    \tilde{\pi}(y) [1 - \alpha_1(y, F_1(y))]  \delta (x - F_2(y)) \alpha_2(y, F_1(y), x).
\end{align*}
Since $F_2 = F_2^{-1}$, the delta distributions are equal if $y=F_2(x)$, allowing us to simplify detailed balance to maintaining
\begin{align*}
    \tilde{\pi}(x) [1 - \alpha_1(x, F_1(x))] \alpha_2(x, F_1(x), y) = \tilde{\pi}(y) [1 - \alpha_1(y, F_1(y))]  \alpha_2(y, F_1(y), x).
\end{align*}
To maximize the acceptance rate under $y=F_2(x)$, we set
\begin{equation*}
    \alpha_2(x, F_1(x), y) = \min \left(1, \frac{\tilde{\pi}(y)}{\tilde{\pi}(x)} \frac{1 - \alpha_1(y, F_1(y))}{1 - \alpha(x, F_1(x))} \right).
\end{equation*}
With access to proposal maps $F_1, F_2$, we can rewrite the acceptance probability from current state $x$ to proposed state $y=F_2(x)$ as
\begin{equation*}
    \alpha_2(x, F_2(x)) = \min \left(1, \frac{\tilde{\pi}(F_2(x))}{\tilde{\pi}(x)} \frac{1 - \alpha_1(F_2(x), F_1(F_2(x)))}{1 \;\; - \;\; \alpha(x, F_1(x))} \right).
\end{equation*}
This formulation highlights the ghost states that emerge from $F_1(F_2(x))$.

\subsubsection{Three Proposal Acceptance Probability}
Let the $\tilde{\pi}$-invariant Markov chain have current state $x$, first rejected state $s_1$, second rejected state $s_2$, and proposed state $y$. We define a non-AC, deterministic proposal kernel with the proposal map $F_3 = P L_{\epsilon_3}^{n_3}$ as
\begin{equation*}
    q_3(x, s_1, s_2, y) = \delta(y-F_3(x))
\end{equation*}
This defines a transition kernel that must account for all four ways to get to state $y$: (1) accepting the first proposal $q_1$ with probability $\alpha_1$; (2) accepting the second proposal $q_2$ with probability $\alpha_2$; (3) accepting the third proposal $q_3$ with some new probability $\alpha_3$; (4) rejecting all proposals and remaining at the current state $y$. Case (2) must marginalize over all second proposals $s_1$ while cases (3) and (4) must marginalize over all second and third proposals $s_1, s_2$, giving us
\begin{align*}
    k(x,y) &= q_1(x,y) \alpha_1(x,y) \\
    &+ \int q_1(x, s_1) [1 - \alpha_1(x, s_1)] q_2(x, s_1, y) \alpha_2(x, s_1, y) ds \\
    &+ \int q_1(x, s_1) q_2(x, s_1, s_2) [ 1- \alpha_1(x, s_1)] [1 - \alpha_2(x, s_1, s_2)] \\
    & \qquad \times [q_3(x, s_1, s_2, y) \alpha_3(x, s_1, s_2, y) + r_3(x) \delta(y)] ds_1 ds_2
\end{align*}
with probability of rejecting the third proposal as $r_3$. Since the first and second terms already satisfied detailed balance and $r_3(x)\delta_x(y)$ is $x \leftrightarrow y$ symmetric, maintaining detailed balance reduces to plugging in deterministic proposal kernels $q_1, q_2, q_3$ and satisfying
\begin{align*}
    \int & \int \tilde{\pi}(x) \delta(s_1 - F_1(x)) [ 1- \alpha_1(x, F_1(x))] \delta(s_2 - F_2(x)) \\
    & \qquad \qquad \qquad \times [ 1 - \alpha_2(x, F_2(x))] \delta(y - F_3(x)) \alpha_3(x, F_3(x)) ds_1 ds_2 = \\
    & \int \int \tilde{\pi}(y) \delta(s_1^\prime - F_1(y)) [ 1 - \alpha_1(y, F_1(y)] \delta(s_2^\prime - F_2(y)) [1 - \alpha_2(y, F_2(y))] \\
    & \qquad \qquad \times \delta(x - F_3(y)) \alpha_3(y, F_3(y)) d s_1^\prime d s_2^\prime
\end{align*}
for dummy variables $s_1, s_2, s_1^\prime, s_2^\prime$. Following the same steps to derive $\alpha_2$ and requiring $y=F_3(x)$, the highest rate of acceptance is
\begin{equation*}
    \alpha_3(x, F_3(x)) = \min \left(1, \frac{\tilde{\pi}(F_3(x))}{\tilde{\pi}(x)} \frac{[1 - \alpha_1(F_3(x), F_1(F_3(x)))]}{1 - \alpha_1(x, F_1(x))} \frac{[1 - \alpha_1(F_3(x), F_2(F_3(x)))]}{1 - \alpha_1(x, F_2(x))} \right).
\end{equation*}
This again reveals the dependencies on ghost points $F_1(F_3(x)), F_2(F_3(x))$.

\subsubsection{\texorpdfstring{$k$}{k}th Proposal Acceptance Probability} Following the same pattern, the general acceptance probability of transitioning from state $x$ to the $k$th proposal $y$ after rejecting $k - 1$ previous proposals is
\begin{equation*}
    \alpha_k(x, F_k(x)) = \min \left(1, \frac{\tilde{\pi}(F_k(x))}{\tilde{\pi}(x)} \prod_{i=1}^{k-1} \frac{1 - \alpha_1(F_k(x), \; F_i(F_k(x)))}{1 \quad  - \quad \alpha_1(x, F_i(x))}  \right).
\end{equation*}
To recover the original acceptance probability in \autoref{eq:accept_prob}, we write out the current state $x$ as $(\theta, \rho^\prime)$, the proposed state $y$ as $(\theta^{\textrm{pr}}, \rho^{\textrm{pr}}) = F_k(\theta, \rho^\prime)$, and rewrite the Gibbs density $\tilde{\pi} \propto \pi(\theta) \text{normal}(\rho \mid 0, M)$ as
\begin{equation*}
    \alpha_k(\theta, \rho^\prime, F_k(\theta, \rho^\prime)) = \min \left(1, \frac{\pi( \theta^{\mathrm{pr}}) \textrm{normal}(\rho^{\mathrm{pr}} \mid 0, M)}{\pi(\theta) \textrm{normal}(\rho^\prime \mid 0, M)} \cdot \prod_{i=1}^{k-1}\frac{ 1 - \alpha_i \big( \theta^{\mathrm{pr}}, \rho^{\mathrm{pr}}, F_i(\theta^{\mathrm{pr}}, \rho^{\mathrm{pr}}) \big)}{ 1 - \alpha_i \big( \theta, \rho^\prime, F_i(\theta, \rho^\prime) \big)} \right).
\end{equation*}

\section{ADDITIONAL NEAL'S FUNNEL EXPERIMENTS}\label{app:funnel}
\subsection{Details of \autoref{fig:funnel_intro}}\label{app:funnel_fig_details}

Each $(x, y_i)$ coordinate of \autoref{fig:funnel_intro} represents a point in parameter space of the $10$ dimensional Neal's funnel density: every $(x,y_i)$ coordinate is transformed into a $10$ dimensional parameter vector $\theta = \{x, y\}$ by repeating $y=\{y_i, \ldots y_i\}$ nine times.

In \autoref{fig:funnel_multiscale} we compute the negative log density and the condition number of its Hessian at each coordinate $\theta$. In \autoref{fig:funnel_drghmc_stepsize}, we initialize DR-G-HMC from each coordinate $\theta$ and generate a single sample. This sample is generated by making sequential proposals attempts with decreasing step sizes. We record the mean step size that generates an \textit{accepted} proposal from $100$ such trials. To avoid proposal rejections, DR-G-HMC uses parameters $\epsilon=2, r=4, K=10$ for many proposal attempts $K$ that reach especially small step sizes $\epsilon_k$; see \autoref{alg:DR-G-HMC}.

\subsection{Neal's funnel in higher dimensions}

We sample from Neal's funnel in $50$, $100$, and $250$ dimensions with $100$ chains each run for $1,000,000$ gradient evaluations. In \autoref{fig:funnel_hist_high_dim} we plot the histogram of the log scale parameter $x$, with known density $x \sim \textrm{normal}(0, 3)$, aggregated across all chains. We observe similar behavior to the $10$ dimensional setting: delayed rejection methods sample deep into the neck of the funnel ($x \ll 0$) while NUTS cannot, despite enormous amounts of compute.

\begin{figure*}[!tbp]
  \centering
  \begin{subfigure}{0.8\textwidth}
    \includegraphics[width=\textwidth]{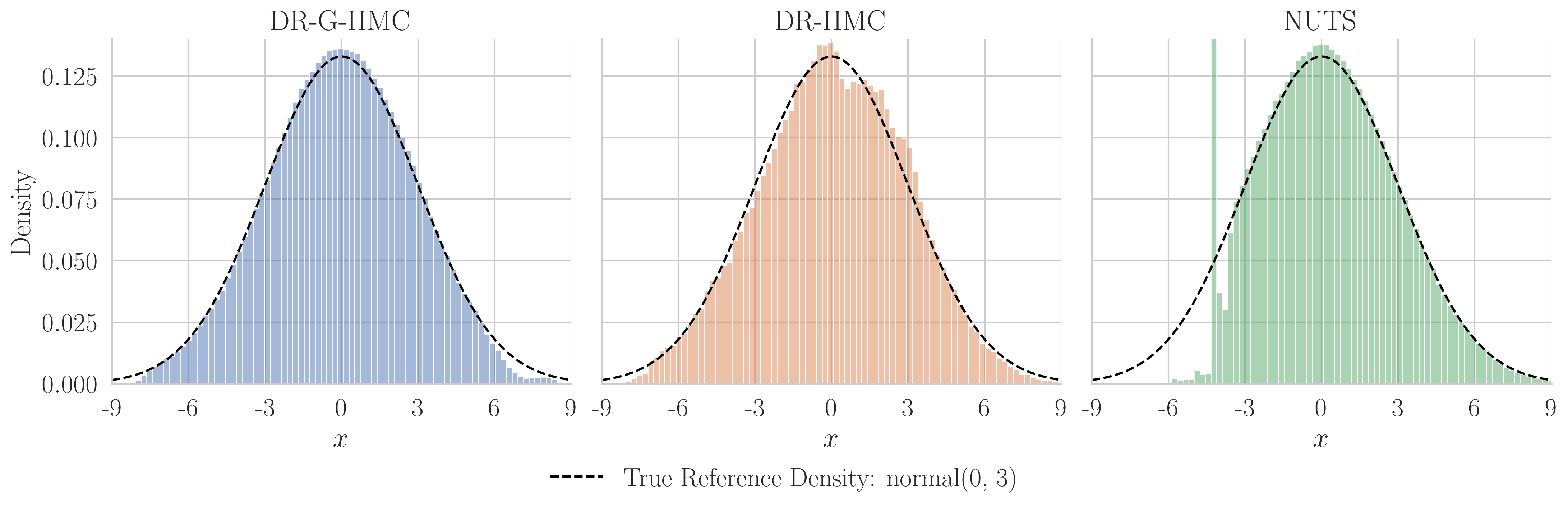}
    \caption{\textbf{$50D$ Neal's funnel.}}
  \end{subfigure}
  \begin{subfigure}{0.8\textwidth}
    \includegraphics[width=\textwidth]{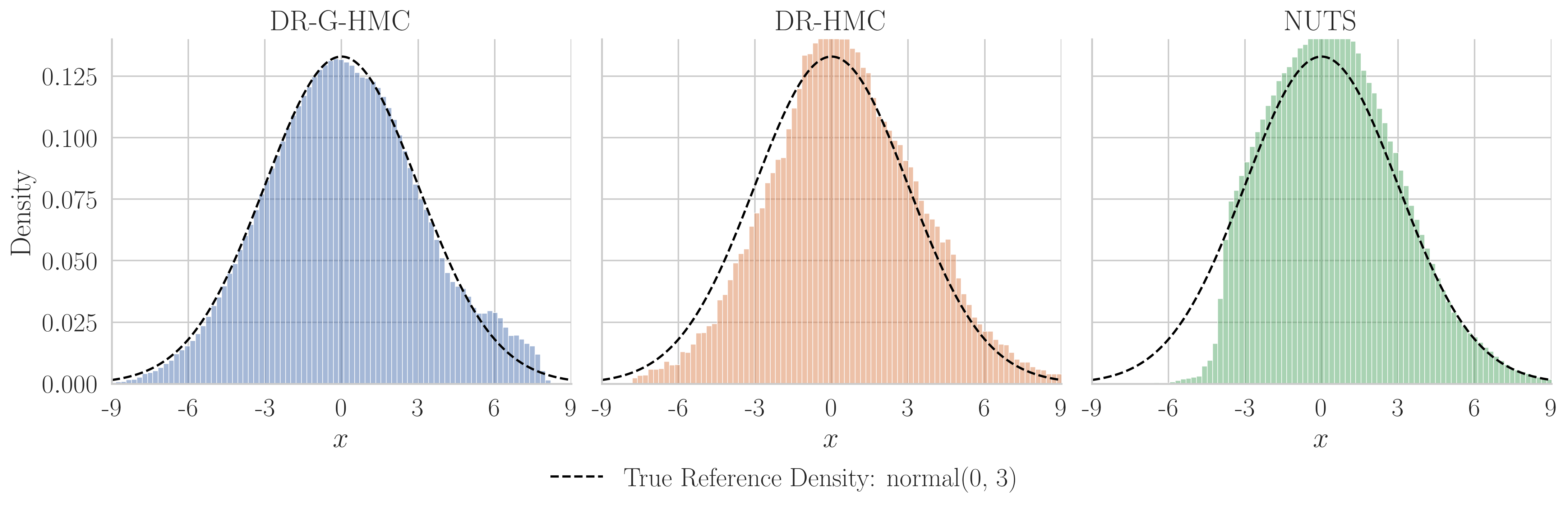}
    \caption{\textbf{$100D$ Neal's funnel.}}
  \end{subfigure}
  \begin{subfigure}{0.8\textwidth}
    \includegraphics[width=\textwidth]{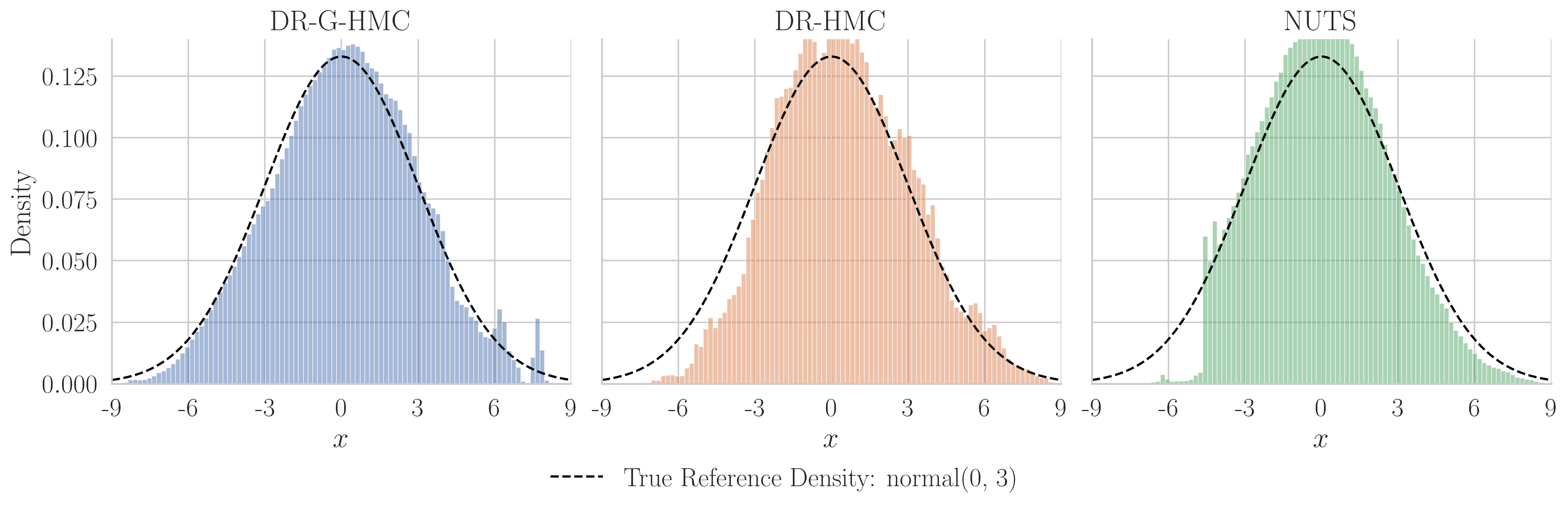}
    \caption{\textbf{$250D$ Neal's funnel.}}
  \end{subfigure}
  \caption{Even in higher dimensions, only delayed rejection methods like DR-G-HMC and DR-HMC can sample deep into the neck of the funnel ($x \ll 0$), while NUTS cannot.}
  \label{fig:funnel_hist_high_dim}
\end{figure*}

\section{TARGET DENSITY DETAILS} \label{app:posterior_exp}

\subsection{Eight Schools}

The eight schools data measures the impact of test-preparation interventions on student test scores across eight different schools \parencite{rubin1981estimation, gelman2021bayesian}. The data contains the average difference in pre-test and post-test scores $y_i$ and standard deviation $\sigma_i$ for the $i$th school. The model parameterizes each school with an efficacy $\theta_i$ drawn from a hierarchical normal prior with unknown location $\mu$ and scale $\tau$. The generative model consists of
\begin{align*}
    \mu \sim \text{normal}(0, 5) \qquad \tau \sim \text{cauchy}_+(0, 5) \\
    \theta_i \sim \text{normal}(\mu, \tau) \qquad y_i \sim \text{normal}(\theta_i, \sigma_i)
\end{align*}
with parameters $\theta = \{\mu, \tau, \theta_1 \ldots \theta_8 \}$. Eight schools is a classic problem in hierarchical Bayesian modeling . It is challenging because $\tau$ dictates the pooling of the treatment effects $\theta_i$ across schools, creating a multiscale, funnel-like geometry between them. 

% We expect, and observe, that delayed rejection methods perform well on this posterior in \autoref{fig:eight_schools_error_vs_grad}.

% \begin{figure}[ht]
%     \centering
%     \includegraphics[width=\linewidth]{figures_appendix/competitive_ghmc/eight_schools.png}
%     \caption{\textbf{Average error vs gradient evaluations for eight schools.} Error in mean ($\mathcal{L}_{\theta, T}$) and variance ($\mathcal{L}_{\theta^2, T}$) averaged over all $100$ chains of sampler draws from eight schools posterior density.}
%     \label{fig:eight_schools_error_vs_grad}
% \end{figure}

\subsection{Normal100}

We consider a posterior over a 100-dimensional normal distribution with mean zero and a dense covariance matrix $\Sigma$ defined by $\Sigma_{ij} = \rho^{\mid i - j \mid}$ for some fixed $\rho \in (0, 1)$. This simulates a covariance matrix with high correlation between adjacent elements that slowly decreases between for distant dimensions $i$ and $j$. The generative model is
\begin{align*}
    \Sigma_{ij} &= \rho^{\mid i - j \mid} \quad i=1 \ldots N, j = 1 \ldots N \\
    y &\sim \text{normal}(0, \Sigma)
\end{align*}
for parameters $\theta = \{ \rho, y_1 \ldots y_N \}$. 

% The error is shown in \autoref{fig:normal100_error_vs_grad}.

% \begin{figure}[ht]
%     \centering
%     \includegraphics[width=\linewidth]{figures_appendix/competitive_ghmc/normal100.png}
%     \caption{\textbf{Average error vs gradient evaluations for normal100.} Error in mean ($\mathcal{L}_{\theta, T}$) and variance ($\mathcal{L}_{\theta^2, T}$) averaged over all $100$ chains of sampler draws from eight schools posterior density.}
%     \label{fig:normal100_error_vs_grad}
% \end{figure}

\subsection{Stochastic Volatility}

Stochastic volatility models seek to model the volatility, or variance, on the return of a financial asset as a latent stochastic process in discrete time  (\cite{kim1998stochastic}). We are given mean corrected returns $y_t$ at $T$ equally spaced time steps and seek to sample the latent log volatility $h_t$, the mean log volatility $\mu$, and persistence of the volatility term $\phi$. The generative model is
\begin{align*}
    \phi \sim \text{uniform}(-1, 1) \qquad &\sigma \sim \text{Cauchy}(0, 5) \qquad \mu \sim \text{Cauchy}(0, 10) \\
    h_1 \sim \text{normal}(\mu, \frac{\sigma^2}{1 - \phi^2}) \qquad &h_t \sim \text{normal}(\mu + \phi(h_{t-1} - \mu), \sigma^2) \quad t = 2, 3, \ldots T \\
   & y_t \sim  \text{normal}(0, e^{h_t}) \quad t = 2, 3, \ldots T
\end{align*}
with parameters $\theta = \{\mu, \sigma, \phi, h_1, \ldots, h_T \}$. The hierarchical prior on the volatility parameters induces strong correlation.

% The error decreases in \autoref{fig:stochastic_volatility_error_vs_grad}.

% \begin{figure}[ht]
%     \centering
%     \includegraphics[width=\linewidth]{figures_appendix/competitive_ghmc/stochastic_volatility.png}
%     \caption{\textbf{Average error vs gradient evaluations for stochastic volatility.} Error in mean ($\mathcal{L}_{\theta, T}$) and variance ($\mathcal{L}_{\theta^2, T}$) averaged over all $100$ chains of sampler draws from stochastic volatility posterior density.}
%     \label{fig:stochastic_volatility_error_vs_grad}
% \end{figure}

\subsection{Item Response Theory} Item response theory models how students answer questions on a test depending on student ability, question difficulty, and the discriminatinative power of the questions (\cite{gelman2006data}). For $I$ students and $J$ questions we are given the binary correctness $y_{ij}$ of student $i$'s answer on question $j$. The model uses the mean-centered ability of the student $i$ as $\alpha_i$, the difficulty of question $j$ as $\beta_j$, and discrimination of question $j$ as $\theta_j$. These terms are combined into something like a logistic regression (but with a multiplicative discrimination parameter). The generative model is
\begin{align*}
    \sigma_\theta &\sim \text{Cauchy}(0, 2) \qquad  \theta \sim \text{normal}(0, \sigma_\theta) \\
    \sigma_a &\sim \text{Cauchy}(0,2) \qquad a \sim \text{lognormal}(0, \sigma_a) \\
    \mu_b &\sim \text{normal}(0, 5) \qquad \sigma_b \sim \text{Cauchy}(0, 2) \qquad b \sim \text{normal}(\mu_b, \sigma_b) \\
    y_i & \sim \text{Bernoulli-Logit}(a_i * (\theta - b_i)) \quad i=1, \ldots, N
\end{align*}
for parameters $\theta = \{ \sigma_\theta, \theta, \sigma_a, a, \mu_b, \sigma_b, b\}$. 
% See error in \autoref{fig:irt_error_vs_grad}.

% \begin{figure}[ht]
%     \centering
%     \includegraphics[width=\linewidth]{figures_appendix/competitive_ghmc/irt_2pl.png}
%     \caption{\textbf{Average error vs gradient evaluations for stochastic volatility.} Error in mean ($\mathcal{L}_{\theta, T}$) and variance ($\mathcal{L}_{\theta^2, T}$) averaged over all $100$ chains of sampler draws from item response theory posterior density.}
%     \label{fig:irt_error_vs_grad}
% \end{figure}

\section{ADDITIONAL SAMPLER EXPERIMENTS} \label{app:samplers}

We provide an extensive comparison of samplers including HMC, G-HMC, DR-HMC, DR-G-HMC, NUTS, and DR-G-HMC with complete momentum refreshment (default DR-G-HMC has \textit{partial} momentum refreshment). Note that DR-G-HMC with complete momentum refresh is equivalent to DR-HMC with one leapfrog step per proposal.

We compute error in mean and variance for these samplers on a variety of multiscale and non-multiscale densities. Recall that multiscale densities are commonly induced by hiearchical probablistic models.

We examine multiscale densities (eight schools and banana) in \autoref{fig:multiscale_densities} and observe that all delayed rejection methods outperform G-HMC, HMC, and NUTS due to their adaptive step size selection.

We examine non-multiscale densities (irt-2pl and normal100) in \autoref{fig:non_multiscale_densities} and observe that DR-G-HMC, unlike DR-HMC, performs comparably to NUTS.

Across both sets of densities, we observe that \textit{partial} momentum refresh (blue) outperforms \textit{complete} momentum refresh (orange) in DR-G-HMC. This indeed suggests that continued motion, facilitated by partial momentum, is key to the performance of DR-G-HMC.

\begin{figure}[H]
    \centering
    \begin{minipage}{\textwidth}
        \centering
        \includegraphics[width=\linewidth]{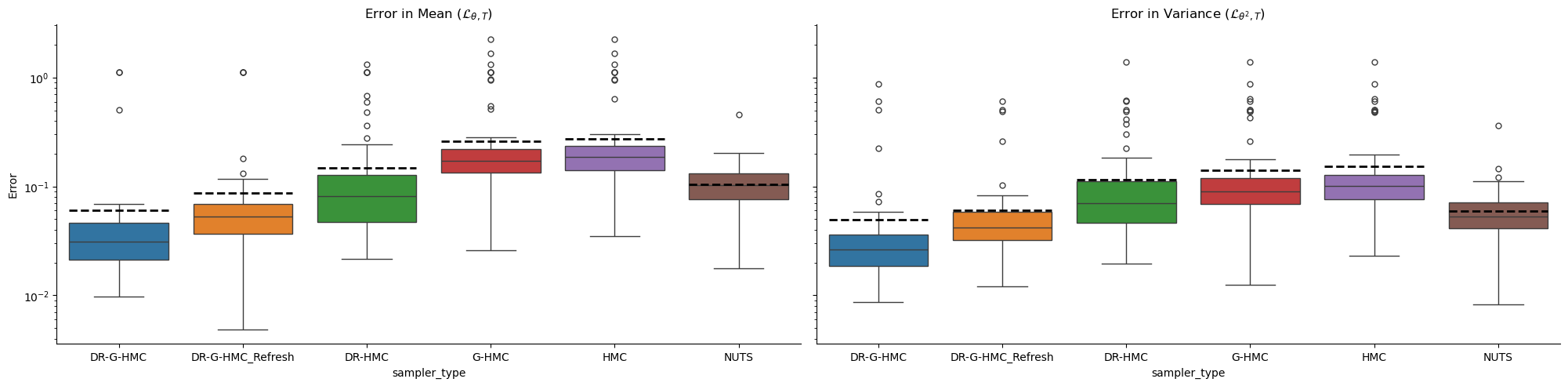}
        \caption*{(a) Eight Schools}
    \end{minipage}
    \vspace{10pt} % Adjust the vertical space as needed
    \begin{minipage}{\textwidth}
        \centering
        \includegraphics[width=\linewidth]{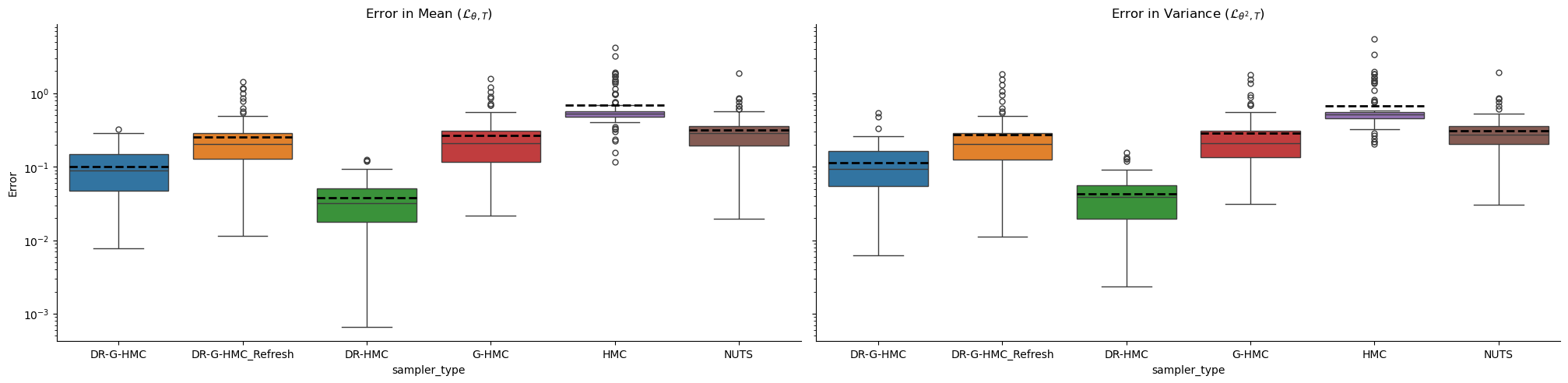}
        \caption*{(b) 2 Dim Banana}
    \end{minipage}
    \vspace{-10pt} % Adjust the vertical space as needed
    \caption{\textbf{Multiscale densities.} Delayed rejection methods achieve low error on multiscale densities.}
    \label{fig:multiscale_densities}
\end{figure}

\begin{figure}[H]
    \centering
    \begin{minipage}{\textwidth}
        \centering
        \includegraphics[width=\linewidth]{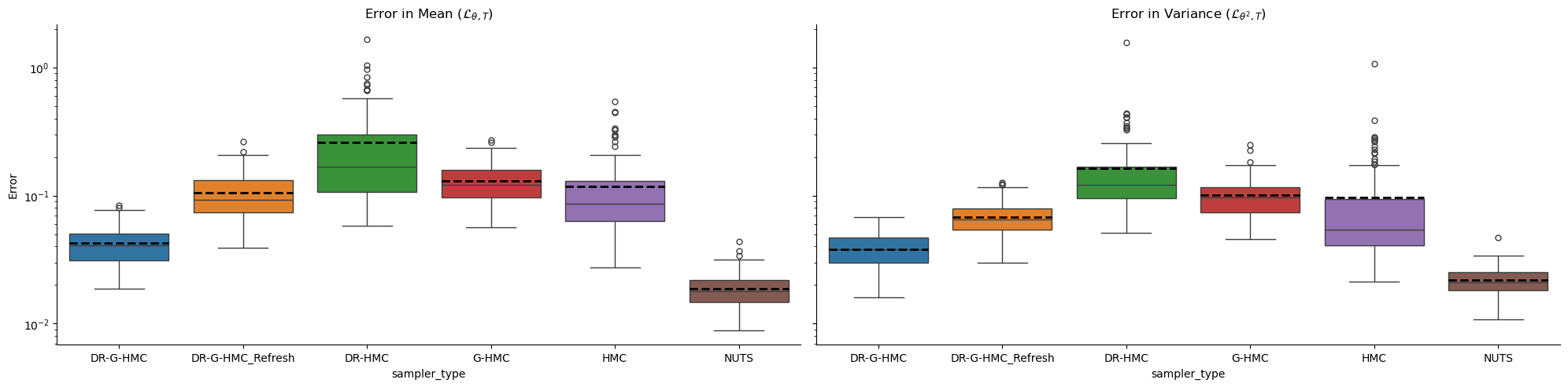}
        \caption*{(a) Normal 100}
    \end{minipage}
    \vspace{10pt} % Adjust the vertical space as needed
    \begin{minipage}{\textwidth}
        \centering
        \includegraphics[width=\linewidth]{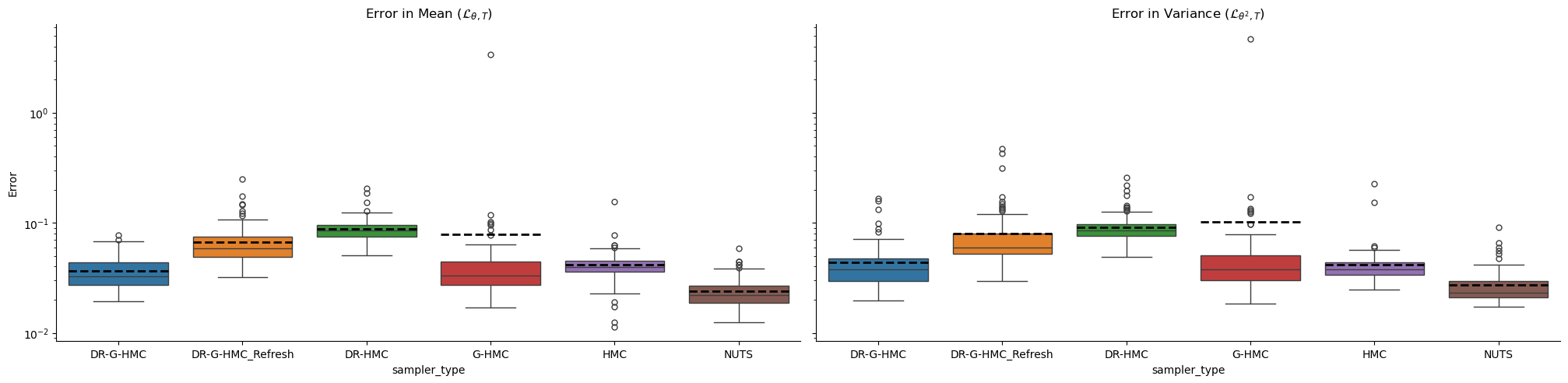}
        \caption*{(b) Item Response Theory}
    \end{minipage}
    \vspace{-10pt} % Adjust the vertical space as needed
     \caption{\textbf{Non-multiscale densities.} DR-G-HMC, unlike DR-HMC, achieves low error on non-multiscale densities.}
    \label{fig:non_multiscale_densities}
\end{figure}

\section{ADDITIONAL TUNING PARAMETER EXPERIMENTS} \label{app:robust}

In this section we investigate how robust DR-G-HMC and DR-HMC are to the tuning of the following hyperparameters: max proposals $K$, step size factor $c$, and reduction factor $r$. See \autoref{sec:drghmc} for details on these hyperparameters.

We run each sampler for $T=10^5$ gradient evaluations on the funnel10, eight schools, and normal100 posteriors and compute the error in mean ($\mathcal{L}_{\theta,T}$) and variance ($\mathcal{L}_{\theta^2,T}$). We find that while both samplers are relatively robust, DR-G-HMC is more robust in multiple settings.

\subsection{Maximum Proposals \texorpdfstring{$K$}{K}}

\begin{figure}[H]
    \centering
    \includegraphics[width=\linewidth]{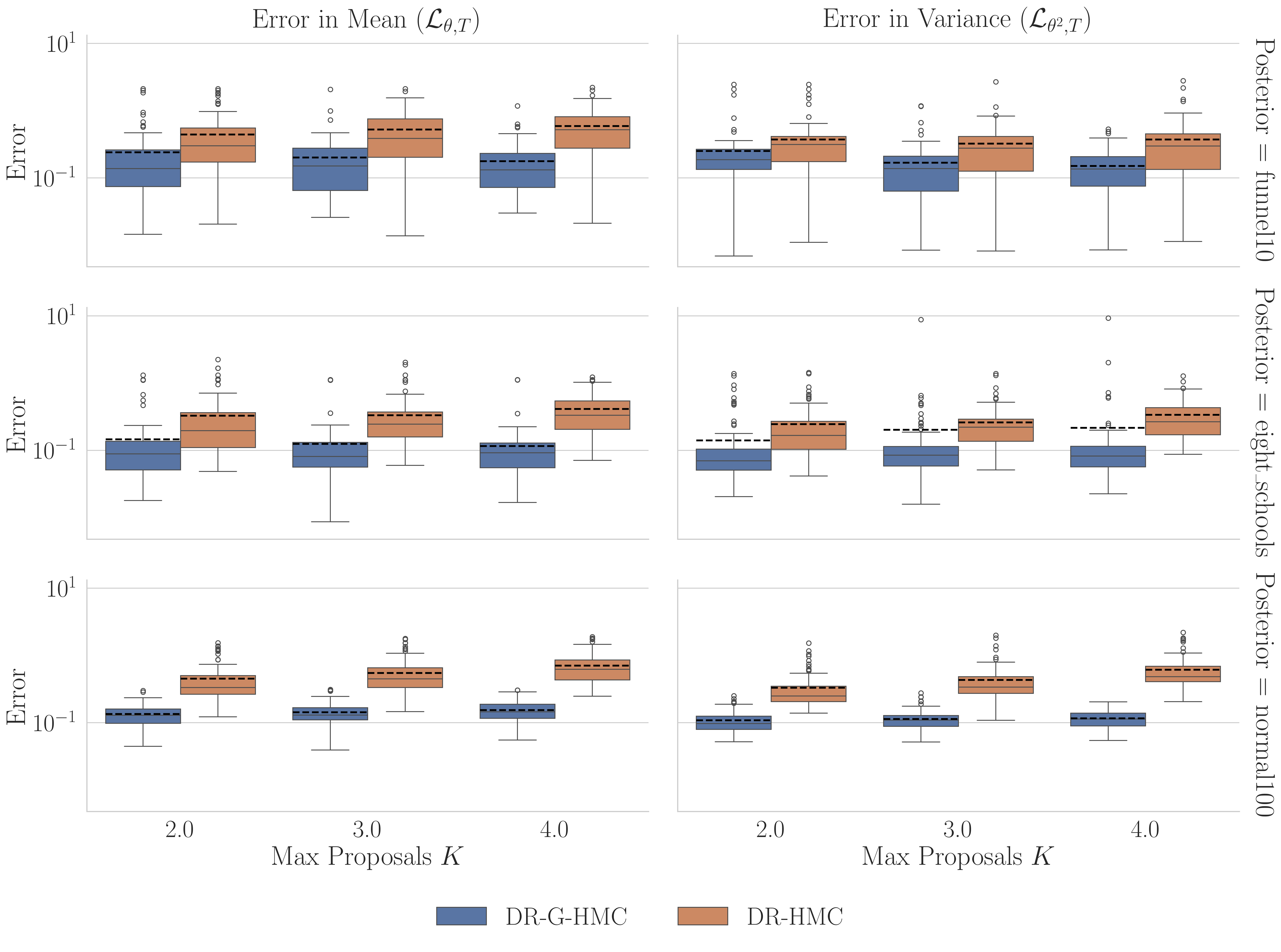}
    \caption{\textbf{DR-G-HMC is robust to the maximum number of proposals attempts $K$.} Error in mean ($\mathcal{L}_{\theta,T}$) and variance ($\mathcal{L}_{\theta^2,T}$) is shown for $100$ chains of various posterior densities. Visual elements represent the following: dashed black line is the mean, solid gray line is the median, colored box is the $(25, 75)$th percentile, whiskers are $1.5$ times the inter-quartile range, and bubbles are outliers.}
    \label{fig:robust_max_proposals}
\end{figure}

We consider maximum proposal attempts $K=[2,3,4]$ in \autoref{fig:robust_max_proposals}. Both DR-G-HMC and DR-HMC errors are mostly constant across different $K$ values. However, the average error across DR-HMC chains (dashed black line) increases with more proposal attempts, especially in the normal100 posterior. DR-G-HMC errors are stable throughout.

\subsection{Reduction Factor \texorpdfstring{$r$}{r}}

\begin{figure}[H]
    \centering
    \includegraphics[width=\linewidth]{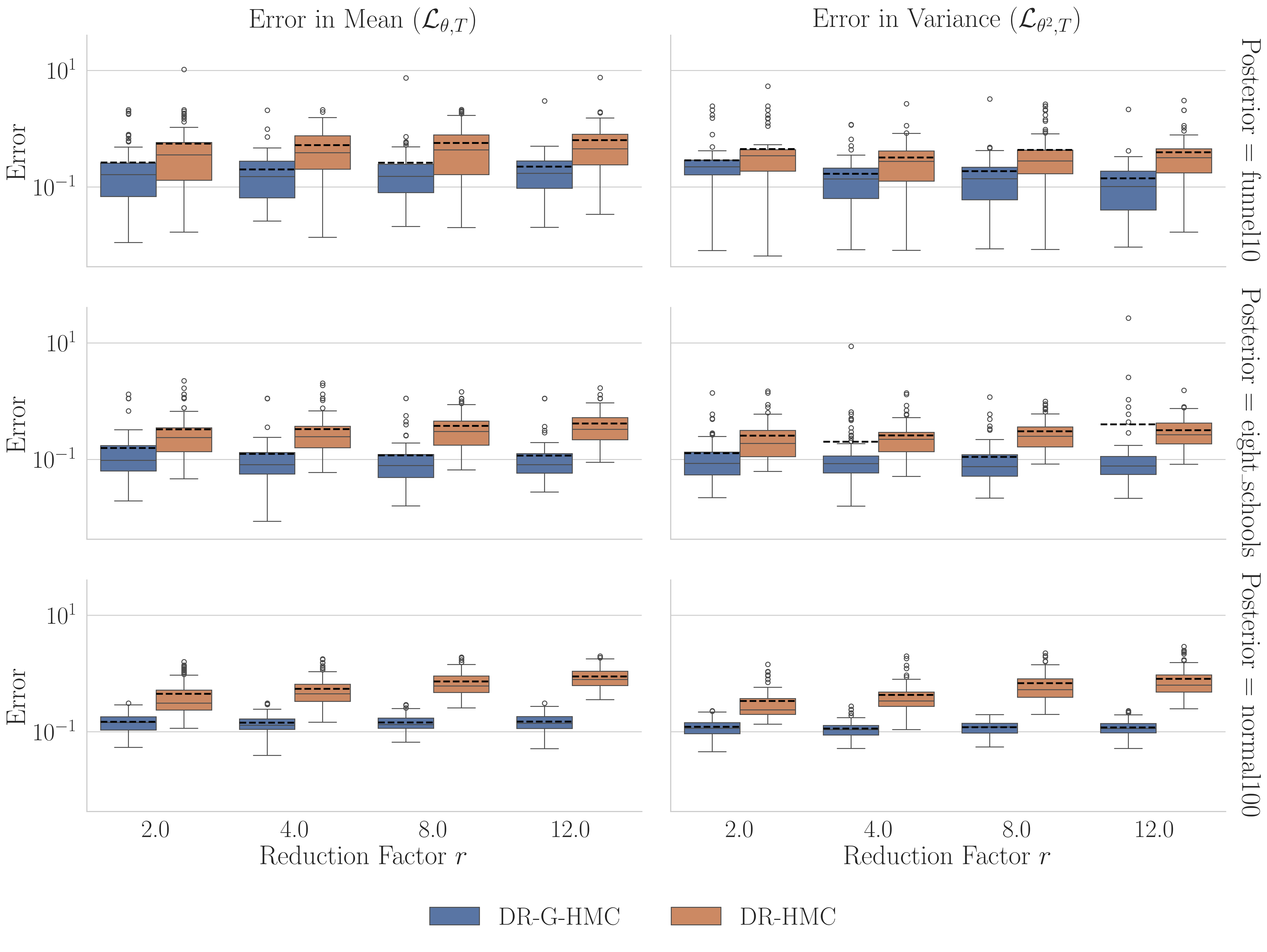}
    \caption{\textbf{DR-G-HMC is robust to tuning of the reduction factor $r$.} Error in mean ($\mathcal{L}_{\theta,T}$) and variance ($\mathcal{L}_{\theta^2,T}$) is shown for $100$ chains of various posterior densities. Visual elements represent the following: dashed black line is the mean, solid gray line is the median, colored box is the $(25, 75)$th percentile, whiskers are $1.5$ times the inter-quartile range, and bubbles are outliers.}
    \label{fig:robust_reduction_factor}
\end{figure}

We consider step size reduction factors $r=[2,4,8,12]$ in \autoref{fig:robust_reduction_factor}. Both DR-G-HMC and DR-HMC errors are relatively stable across different $r$ values. However, the average error across DR-HMC chains (dashed black line) increases with more proposal attempts in the normal100 posterior. DR-G-HMC performance is constant everywhere except for variance of the eight schools posterior. For $r=4,12$, DR-G-HMC has a single outlier that dramatically increase its average error across chains (dashed black line).

\subsection{Step Size Factor \texorpdfstring{$c$}{c}}

\begin{figure}[H]
    \centering
    \includegraphics[width=\linewidth]{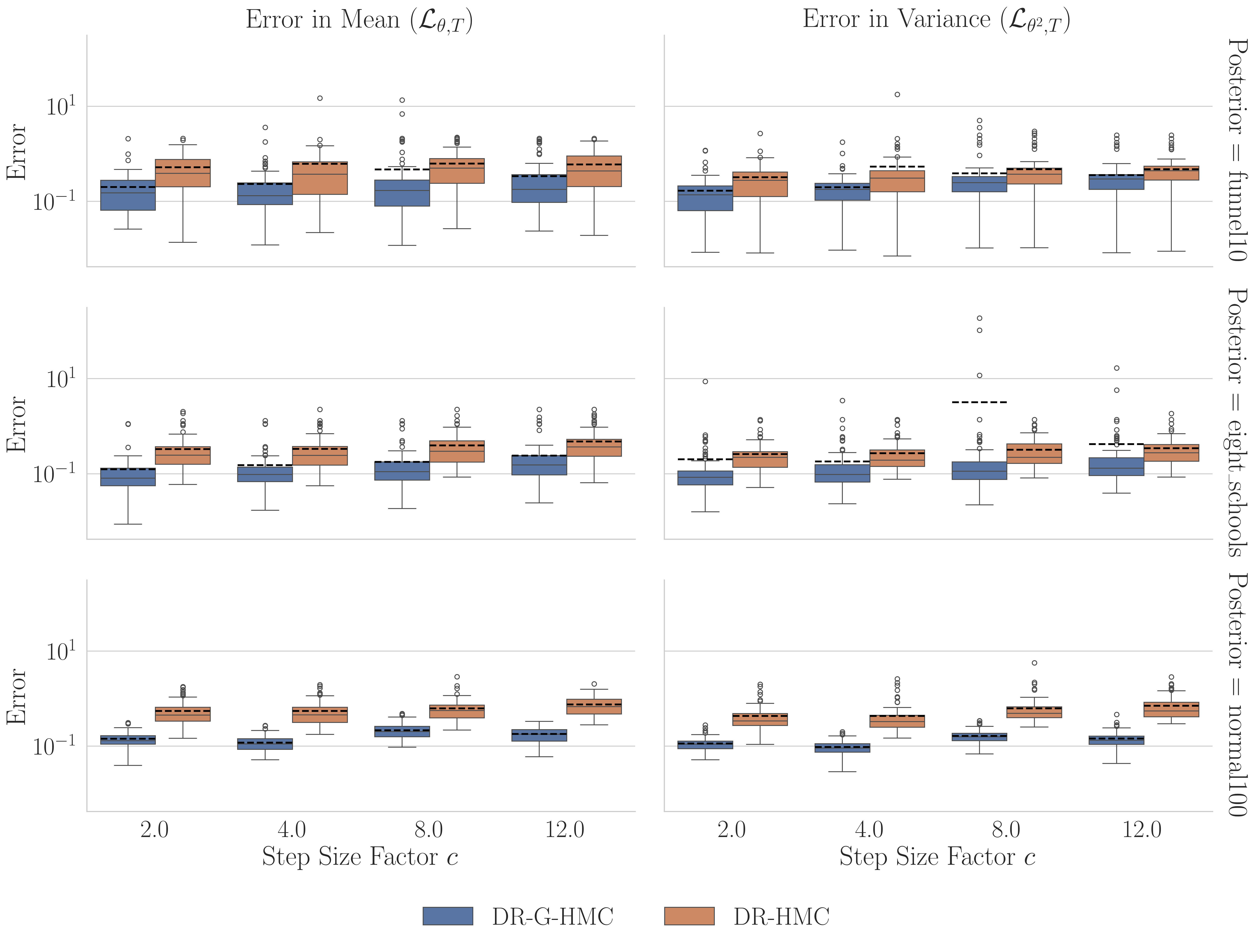}
    \caption{\textbf{DR-G-HMC is robust to tuning of the reduction factor $r$.} Error in mean ($\mathcal{L}_{\theta,T}$) and variance ($\mathcal{L}_{\theta^2,T}$) is shown for $100$ chains of various posterior densities. Visual elements represent the following: dashed black line is the mean, solid gray line is the median, colored box is the $(25, 75)$th percentile, whiskers are $1.5$ times the inter-quartile range, and bubbles are outliers.}
    \label{fig:robust_step_size_factor}
\end{figure}

We consider the step size factor $c=[2,4,8,12]$ in \autoref{fig:robust_step_size_factor} that is used to compute the initial step size $\epsilon = c \epsilon_{\textrm{ NUTS}}$ for  NUTS's adapted step size $\epsilon_{\textrm{ NUTS}}$. The average error across DR-HMC chains (dashed black line) increases with a larger step size factor $c$ in the normal100 posterior. DR-G-HMC performance is constant everywhere except for variance of the eight schools posterior. For $c=8$, DR-G-HMC has a few outliers that dramatically increase its average error across chains (dashed black line).

\section{WALLCLOCK RUNTIME OF SAMPLERS} \label{app:runtime}

Profiling experiments reveal the runtime of our DR-G-HMC sampler is trivially longer than the runtime of DR-HMC. The overhead of DR-G-HMC arises from resampling the momentum at every iteration (not from computing the recursive acceptance probability).

For a fair comparison, we investigate $n$ iterations of DR-G-HMC and a single iteration of DR-HMC with  $n$ leapfrog steps. In this setup, both samplers compute $n$ gradients of the log density $\nabla \log p(x)$. Because the gradient is computed with automatic-differentiation, we get the log density $\log p(x)$ for free -- this is later used to cheaply compute the acceptance probability.

However, DR-G-HMC resamples the momentum $\rho$ from a normal distribution $n$ times while DR-HMC only does so once. This is the source of the minor runtime differences between DR-G-HMC and DR-HMC.

The exact slowdown of DR-G-HMC depends on the relative amount of time between resampling the momentum $\rho \sim \textrm{normal}(\cdot, \cdot)$ vs computing the gradient of the log density $\nabla \log p(x)$. This is impacted by the number of leapfrog steps $n$ and how long it takes to compute the gradient of the log density $\nabla \log p(x)$. The latter involves differentiating through the prior and likelihood, which can be arbitrarily complex.

\autoref{tab:runtime} shows runtime for DR-G-HMC is marginally slower than DR-HMC for Neal's Funnel and IRT-2PL target densities. However, as discussed in the main paper, DR-G-HMC outperforms DR-HMC when sampling from non-multiscale densities because of its algorithmic advantages. We thus expect that on most real-world densities DR-G-HMC will achieve a low error faster than DR-HMC.

\begin{table*}
\centering
\begin{tabular}{c c c}
\toprule
     & \textbf{Neal's Funnel ($d=250$)} & \textbf{IRT-2PL ($d=144$)} \\
    \midrule
    \textbf{DR-G-HMC} & $0.0012 (\pm 0.0003)$ & $0.0019 (\pm 6e-05)$ \\
    \textbf{DR-HMC} & $0.0009 (\pm 0.0037)$ & $0.0013 (\pm 1.6e-4)$ \\
\bottomrule \\
\end{tabular}
\caption{\textbf{Wallclock runtime of DR-G-HMC and DR-HMC samplers.} Mean and standard deviation of a single iteration are reported, averaged over $1000$ trials and measured in seconds.}
\label{tab:runtime}
\end{table*}

\end{document}